\documentclass[sigconf,nonacm]{acmart}
\setcopyright{none}
\usepackage{comment}
\usepackage[linesnumbered,ruled,vlined]{algorithm2e}

\usepackage{algorithmicx}
\usepackage[noend]{algpseudocode}
\usepackage{graphicx}
\usepackage{textcomp}
\usepackage{xcolor}
\usepackage{booktabs}

\usepackage{threeparttable}

\newtheorem{theorem}{Theorem}
\newtheorem{definition}{Definition}
\newtheorem{problem}{Problem}
\usepackage{ulem}
\usepackage{pifont}

\def\BibTeX{{\rm B\kern-.05em{\sc i\kern-.025em b}\kern-.08em
    T\kern-.1667em\lower.7ex\hbox{E}\kern-.125emX}}
\begin{document}

\title{SparStencil: Retargeting Sparse Tensor Cores to Scientific Stencil Computations via Structured Sparsity Transformation$^{\ddagger}$} 

\author{Qi Li}
\authornote{Work done during an internship at Microsoft Research (May 1, 2024)}
\affiliation{%
  \institution{University of Science and Technology of China}
    \country{China}
}

\author{Kun Li}
\authornote{Corresponding author: \texttt{likungw@gmail.com}\\
$^{\ddagger}$\textbf{Accepted to SC'25 (June 3, 2025). This work was previously submitted to ISCA’25 (Nov 22, 2024) and substantially revised based on feedback.}}
\affiliation{%
  \institution{Microsoft Research}
    \country{China}
}
\author{Haozhi Han}
\affiliation{%
  \institution{Peking University}
    \country{China}
}
\author{Liang Yuan}
\affiliation{%
  \institution{Chinee Academy of Sciences}
    \country{China}
}

\author{Junshi Chen}
\affiliation{%
  \institution{University of Science and Technology of China}
    \country{China}
}

\author{Yunquan Zhang}
\affiliation{%
  \institution{Chinese Academy of Sciences}
    \country{China}
}

\author{Yifeng Chen}
\affiliation{%
  \institution{Peking University}
    \country{China}
}

\author{Hong An}
\affiliation{%
  \institution{University of Science and Technology of China}
    \country{China}
}

\author{Ting Cao}
\affiliation{%
  \institution{Tsinghua University}
    \country{China}
}

\author{Mao Yang}
\affiliation{%
  \institution{Microsoft Research}
    \country{China}
}

\begin{abstract}
Sparse Tensor Cores offer exceptional performance gains for AI workloads by exploiting structured 2:4 sparsity. However, their potential remains untapped for core scientific workloads such as stencil computations, which exhibit irregular sparsity patterns.

This paper presents SparStencil, the first system to retarget sparse TCUs for scientific stencil computations through structured sparsity transformation. SparStencil introduces three key techniques:
(1) Adaptive Layout Morphing, which restructures stencil patterns into staircase-aligned sparse matrices via a flatten-and-crush pipeline;
(2) Structured Sparsity Conversion, which formulates transformation as a graph matching problem to ensure compatibility with 2:4 sparsity constraints;
(3) Automatic Kernel Generation, which compiles transformed stencils into optimized sparse MMA kernels via layout search and table-driven memory mapping.
Evaluated on 79 stencil kernels spanning diverse scientific domains, SparStencil achieves up to 7.1x speedup (3.1x on average) over state-of-the-art framework while reducing code complexity and matching or exceeding expert-tuned performance in both compute throughput and memory efficiency.
%Sparse Tensor Core Units (TCUs) have emerged as a critical accelerator for sparse matrix multiplication in AI applications, efficiently handling structured sparsity to boost computational performance. However, due to its strict requirements for structured sparsity, the potential for improving critical scientific operations like stencil computations remains untapped.

%This paper presents SparStencil, a novel stencil computing system designed to unlock the capabilities of sparse Tensor Cores for stencil computations in scientific computing through Permutation Invariant Transformation. SparStencil incorporates three core techniques: (1) Adaptive Layout Morphing with stencil flattening and duplicates crush. (2) Structured Sparsity Conversion via Permutation Invariant Transformation. (3) Automatic Kernel Generation with optimal matrix layout transformation and performance-boosting optimizations. SparStencil outperforms other stencil optimization frameworks, achieving significant speedups compared to solutions like AMOS, cuDNN, Brick, DRStencil, TCStencil and ConvStencil. As the first work to unlock sparse tensor cores for stencil computations, SparStencil promises to boost the performance of more scientific and engineering applications beyond AI workloads.
\end{abstract}

\maketitle

%\begin{IEEEkeywords}
%component, formatting, style, styling, insert
%\end{IEEEkeywords}

%Qi Li
%Kun Li
%Hong An
%Yunquan Zhang
%Ting Cao
%Mao Yang

\section{Introduction}
\label{se:intro}
The rapid ascent of artificial intelligence (AI) workloads has fundamentally reshaped the computing landscape. Modern accelerators—especially GPUs—are now engineered not for general-purpose computing, but to saturate the needs of AI-native operations, such as matrix multiplications with structured sparsity. At the heart of this shift are Sparse Tensor Cores, which exploit fine-grained sparsity patterns (e.g., 2:4 sparsity) to skip redundant zero-element computations and significantly boost performance.

Yet, scientific computing has largely remained on the periphery of this acceleration revolution. While sparse TCUs have become a cornerstone of AI performance, their capabilities remain underutilized in scientific applications, where data and compute structures diverge sharply from those in AI. This disconnect reflects a deeper tension: hardware is evolving toward AI-native execution models, but science remains anchored in patterns that predate them.

Nowhere is this divergence more evident than in stencil computations, the backbone of applications such as fluid dynamics~\cite{HUYNH2014209, LUSHER2021108063}, weather prediction~\cite{AoIPDPS, Ben-NunSC}, and earth modeling~\cite{MathiasSC22}. These computations iteratively update grid values based on local neighborhoods and account for over 75\% of the runtime in long-running simulations~\cite{denzler2023casperacceleratingstencilcomputation}. Their ubiquity and cost have earned them recognition as one of the “seven dwarfs” of scientific computing~\cite{asanovic2006landscape, 10.1145/1562764.1562783, 10.1145/3458817.3476154}.

However, stencil computations yield clustered, irregular sparsity patterns that fundamentally differ from the dense or uniformly sparse structures seen in AI. When mapped into matrix formats for GPU acceleration, these patterns often lead to dense masking, padding, and inefficient memory access—all of which violate the alignment constraints of sparse TCUs and degrade utilization. Prior approaches have largely focused on suppressing this sparsity through dense emulation, missing the opportunity to align it with the execution model of sparse TCUs.

This paper asks a fundamental question:
Can we reimagine scientific computing not as a victim of hardware divergence but as a domain whose latent structure can be transformed to align with AI-native accelerators?

We argue yes. And we answer with SparStencil, a stencil computing system that retargets sparse tensor cores originally designed for AI workloads to accelerate scientific stencil computations. 
Instead of coercing stencil computations into dense formats, SparStencil embraces their clustered sparsity, restructures it into 2:4-compatible formats, and compiles it into optimized sparse TCU kernels—automatically and portably.

SparStencil is guided by four key insights that reshape how stencil sparsity can be efficiently exploited on sparse tensor cores.
(1) Sparsity is not overhead—it’s alignment potential.  Prior approaches treated stencil-induced sparsity as overhead to be suppressed in order to fit dense compute units. Instead, we view this sparsity as an alignment opportunity for specialized sparse execution. 
(2) Residual sparsity is a leverage point. Despite masking and padding, stencil workloads retain high residual sparsity (50–80\%). These conditions are poorly handled by dense TCUs but can be harnessed by sparse TCUs once the structure is exposed.
(3) The irregularity does not mean a lack of structure. Though irregular at first glance,  stencil-induced sparsity exhibits consistently clustered spatial patterns. These can be systematically reshaped into alignment-friendly forms with the right abstraction and transformation.
(4) Structure-aware compilation closes the loop.
Exploiting sparse TCUs demands more than format conversion—it requires compilers that understand and reason over structure, automating the path to performance.

To realize these insights, SparStencil consists of three tightly integrated stages:

First, Adaptive Layout Morphing reformulates diverse stencil patterns into unified matrix-matrix representations through a two-step flatten-and-crush pipeline. It systematically eliminates redundant data accesses while exposing staircase-aligned sparsity, a regularized layout pattern that serves as the structural backbone for downstream hardware-oriented transformations.

Second, Structured Sparsity Conversion transforms staircase-aligned sparsity into 2:4-compatible formats required by sparse TCUs. It is achieved by formulating a column-wise conflict graph, where nodes represent matrix columns and edges encode alignment violations. A hierarchical two-level matching algorithm is then applied to resolve conflicts, ensuring valid 2:4 alignment while minimizing zero-padding and maximizing TCU utilization.

Third, Automatic Kernel Generation compiles the transformed computation into high-performance sparse MMA kernels. It integrates layout exploration, table-driven memory mapping, and code synthesis into a fully automated backend, delivering expert-level performance without manual tuning or architecture-specific reengineering.

We evaluate SparStencil across 79 real-world stencil kernels drawn from a wide spectrum of application domains, including partial differential equation solvers, fluid dynamics, lattice Boltzmann methods, phase field models, and geophysical simulations. Compared to state-of-the-art baselines, SparStencil achieves an average 3.1x speedup, peaking at 7.1x, while significantly reducing code complexity. Our automatic compiler approach matches or exceeds manually optimized expert code in both compute throughput and memory efficiency.

To our knowledge, SparStencil is the first system to map scientific stencil computations onto sparse AI accelerators—not by tuning low-level kernels but by structurally rethinking how stencil workloads are represented and executed. This transformation retargets sparse Tensor Cores to scientific domains, and lays the groundwork for a broader shift toward accelerator-centric design for scientific computing.

\section{Background and Challenges}
\label{se:background}
\subsection{Sparse Tensor Core}
\label{se:TCU}

Sparse Tensor Core is a cutting-edge hardware architecture specifically designed to accelerate sparse matrix multiplication, as described in Equation \ref{eq:mma}. 
\begin{equation}
D_{m \times n} = (A_{m \times k} \odot M_{m \times k}) \times B_{k \times n} + C_{m \times n}
\label{eq:mma}
\end{equation}
where $M_{m \times k}$ is a binary mask matrix in which each group of 4 consecutive elements in a row contains exactly 2 elements equal to 1 and 2 elements equal to 0. This constraint, referred to as 2:4 sparsity, is represented as follows:
\begin{equation}
    M_{i, 4j + l - 4} \in \{0, 1\} \quad \text{and} \quad \sum_{l=1}^{4} M_{i, 4j + l - 4} = 2
\label{eq:sparsity}
\end{equation}
While sparse TCUs strictly enforce this 2:4 sparsity pattern, sub-2:4 patterns (0:4 and 1:4) can still be processed. This is achieved by treating one or two elements in the 4-element block as nonzero, even if they are zero. Since multiplying by zero does not affect the result, this approach preserves computational correctness. 
 
At the hardware level, there is an additional constraint; sparse TCUs partition matrices into uniformly sized fragments for computation. Regardless of matrix size variations, these fragments remain fixed (e.g., $16 \times 16 \times 8$, $16 \times 32 \times 8$).~\cite{pool2021sparsity}

% Table \ref{tab:performance_comparison} delineates the computational precisions and corresponding fragment dimensions facilitated by sparse TCUs~\cite{pool2021sparsity, nvidia_ptx_isa_8.5}.

Despite the stringent 2:4 sparsity constraints and invariant fragment size limitations imposed by sparse TCUs, it offers substantial computational potential. Across all supported data types (FP16, BF16, and TF32), sparse TCUs achieve up to 32x speedup over traditional Fused Multiply-Add (FFMA) operations and consistently deliver 2x the performance of dense TCUs, presenting a substantial opportunity to accelerate fundamental operators in traditional scientific computing domains, such as stencil computations.~\cite{nvidia_ptx_isa_8.5}

% Table \ref{tab:performance_comparison} highlights the impressive performance gains of sparse TCU compared to standard Fused Multiply-Add (FFMA) operations and dense TCU. It demonstrates that sparse TCU achieves twice the speed of dense TCUs across all precision levels,
% \begin{table}[h]
%   \centering
%   \caption{The Superior Performance of Sparse Tensor Cores}
%   \label{tab:performance_comparison}
%   \small
%   % \begin{tabular}{p{1cm}p{1.7m}p{1.7cm}}
%   \begin{tabular}{cccc}
%     \toprule
%     \textbf{Input} & \textbf{TCU Fragment} &  \textbf{vs. FFMA} & \textbf{vs. Dense TCUs} \\
%     \midrule
%     FP16 / BF16 & $16 \times 16 \times 8$ & 32X & 2X \\
%     FP16 / BF16 & $16 \times 32 \times 8$ & 32X & 2X \\
%     TF32        & $16 \times 16 \times 8$ & 16X & 2X \\
%     TF32        & $16 \times 8  \times 8$ & 16X & 2X \\
%     % INT8        & $-$                     & 64X & 2X \\
%     \bottomrule
%   \end{tabular}
% \end{table}

\subsection{Stencil Computation}

Stencil computations play a critical role in numerous scientific and engineering applications, which are extensively involved in various domains from physical simulations to machine learning. Stencil is also included as one of the seven computational dwarfs presented in the Berkeley View and arises as a principal class of floating-point kernels in high-performance computing.

A stencil contains a pre-defined pattern that updates each point in a $d$-dimensional spatial grid iteratively along the time dimension. These patterns are typically categorized into two main types: star and box stencils. In a star stencil, the computation involves a weighted sum of the central point and its neighboring points along individual dimensions. Conversely, a box stencil computes the weighted sum over a neighborhood that forms a square or cube around the central point in 2D or 3D, respectively. 

% Algorithm 1 demonstrates the Box-2D9P stencil. (radius: 1, shape: box, dimensions: 2, points: 9).

% \begin{algorithm}
% \caption{Box-2D9P Stencil}
% \KwIn{mesh A, weight $c_{11} \sim c_{33}$, number of time steps $T$}
% \KwOut{mesh B}
% \For{time step $t = 1$ to $T$}{
%     \For{point $[i][j]$ in A}{
%         $B[i][j] = c_{11} \times A[i-1][j-1] + c_{12} \times A[i-1][j] + c_{13} \times A[i-1][j+1]$ \\
%         $+ c_{21} \times A[i][j-1] + c_{22} \times A[i][j] + c_{23} \times A[i][j+1]$ \\
%         $+ c_{31} \times A[i+1][j-1] + c_{32} \times A[i+1][j] + c_{33} \times A[i+1][j+1]$
%     }
%     Swap mesh A and mesh B
% }
% \end{algorithm}

\begin{figure}[htbp]
\centering
\includegraphics[width=0.49\textwidth]{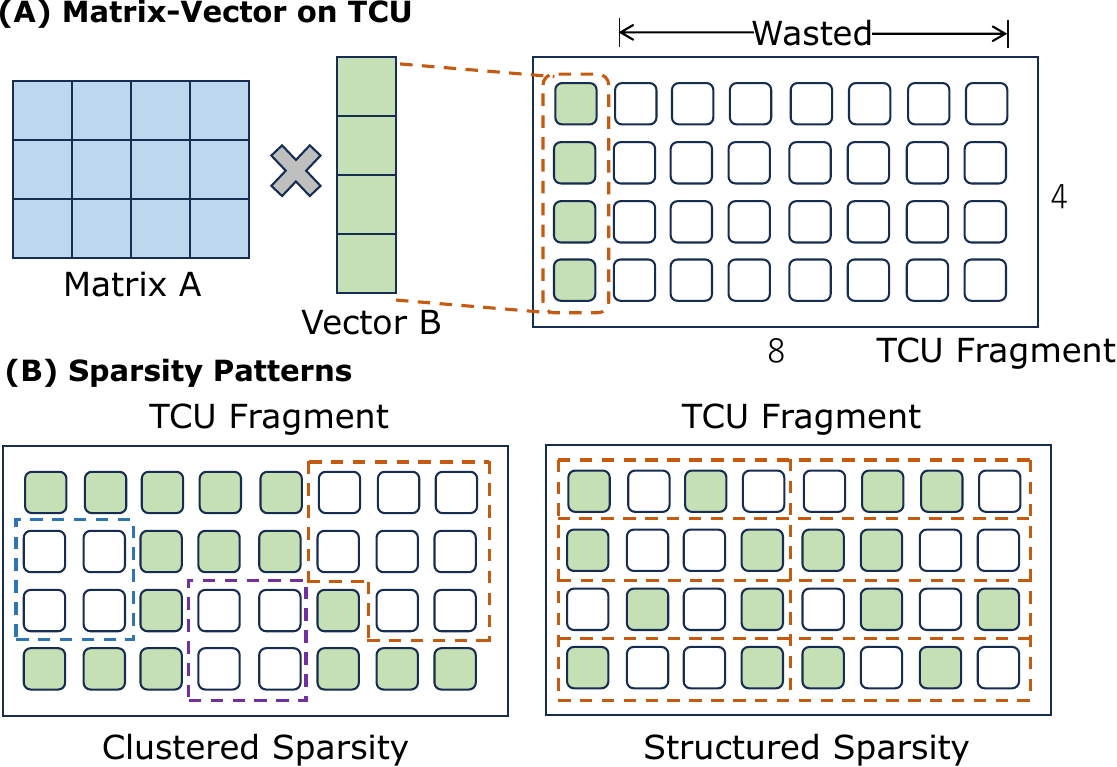} % Adjust width here
\caption{Architectural mapping challenges of stencil computations on TCUs. (a) Naive matrix-vector mapping achieves only $12.5\%$ compute utilization ($\frac{1}{8}$ active columns) on an $8 \times 4$ TCU fragment due to rigid dataflow constraints. (b) The clustered sparsity and structured sparsity in TCUs.}
\label{fig:fragment}
\end{figure}
\subsection{Divergence and Opportunities}
\label{se:opp}

Sparse TCUs have provided dedicated acceleration for the most fundamental matrix multiplication operations in the AI field. However, essential scientific operators such as stencil, face significant challenges in effectively harnessing AI-specific accelerators like TCUs, despite its immense computational potential being a highly valuable yet underutilized resource.

% Fortunately, some efforts have already been made to bridge this divergence by mapping stencil computations onto dense TCU, such as TCStecnil and ConvStencil. However, these adaptations inevitably introduce more than $50\%$ \textit{unstructured sparsity}, leading to underutilization of more than half of the TCU's computational throughput. The \textit{unstructured sparsity}, as illustrated in Figure \ref{fig:fragment} refers to sparsity patterns in a micro-tile---composed of four consecutive horizontal elements in a matrix---that do not adhere to the 2:4 sparsity requirement of sparse TCUs.

% Sparse TCU has provided dedicated acceleration for the most fundamental matrix multiplication operations in the AI field. However, \qi{despite its immense computational potential being a highly valuable yet underutilized resource}, essential scientific operators such as stencil, face significant challenges in effectively harnessing AI-specific accelerators like TCU.

For instance, a direct translation of stencil computations to naive matrix-vector multiplication via a canonical im2row-style approach exhibits severely low accelerator utilization. As shown in Figure \ref{fig:fragment}(a), $87.5\%$ columns are wasted due to the fixed fragment dimensions inherent in TCU architectures.

Fortunately, some efforts have already been made to bridge this divergence by mapping stencil computations onto dense TCUs, such as TCStencil and ConvStencil. However, these adaptations inevitably introduce more than $50\%$ \textit{clustered sparsity}, leading to the underutilization of more than half of the TCU's computational throughput. The \textit{clustered sparsity}, as illustrated in Figure \ref{fig:fragment}(b), refers to sparsity patterns containing contiguous regions with excessive zero concentrations that violate the structured sparsity requirements (2:4 sparsity) of sparse TCUs.

% \kun{The \textit{clustered sparsity}, as illustrated in Figure \ref{fig:fragment}(b) refers to sparsity patterns in a micro-tile---composed of four consecutive horizontal elements in a matrix---that do not adhere to the 2:4 sparsity requirement of sparse TCUs.}

While clustered sparsity has traditionally been seen as a barrier to dense TCUs, it offers a unique opportunity in the context of sparse TCUs. If we can convert the clustered sparsity of the matrix into structured 2:4 sparsity, we can seamlessly integrate sparse TCUs into stencil computations. Instead of focusing on eliminating zero elements, this approach enables us to leverage the hardware to significantly increase computational density. Consequently, rather than viewing sparsity as a limitation, we propose that it should be seen as an opportunity to harness the capabilities of sparse TCUs, unlocking new pathways for stencil computations.

% However, PIT offers new opportunities by rearranging the distribution of sparsity within the matrix without altering the final computational result. If we can use PIT to convert the clustered sparsity of the matrix into structured 2:4 sparsity, we can seamlessly integrate sparse TCUs into stencil computations. Instead of focusing on eliminating zero elements, this approach enables us to strategically leverage the hardware to significantly increase computational density. Consequently, rather than viewing sparsity as a limitation, we propose that it should be seen as an opportunity to harness the capabilities of sparse TCUs, unlocking new pathways for enhanced computational performance.
\begin{figure}[htbp]
\centering
\includegraphics[width=0.48\textwidth]{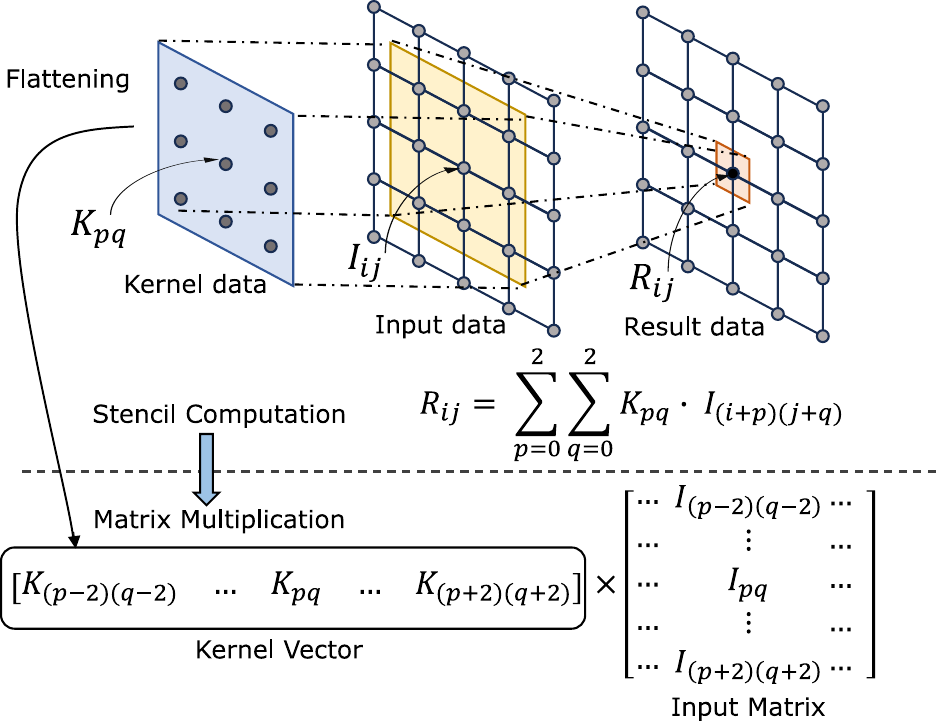}
\caption{Example of the stencil flattening process.}
\label{fig:flattening}
\end{figure}
\subsection{Challenges}

Despite the opportunity, systematically and efficiently mapping stencil computations with various stencil kernels onto sparse TCUs is not a straightforward task. In this subsection, we outline and discuss three key challenges: 

Firstly, \textit{how can we map stencil kernels onto TCUs with low redundancy and sparsity-awareness?} The matrix-based transformation of stencil computations inherently introduces extensive duplication, as sliding-kernel overlap causes each input element to appear repeatedly across the matrix, leading to significant memory overhead. Moreover, existing works (e.g., TCStencil and ConvStencil) lack sparsity awareness, resulting in irregular patterns that cannot be aligned with sparse TCUs. Lacking the ability to restructure such sparsity into a usable form, they often choose to ignore it entirely and fall back to dense computation.

% Firstly, \textit{how can we map varying stencil kernels onto TCUs with low redundancy and sparsity-awareness?} The constraints imposed by fixed TCU fragment sizes hinder the formulation of a unified and efficient mapping strategy for various stencil kernels, leading existing works (e.g., TCStencil and ConvStencil) to rely predominantly on case-by-case solutions. Moreover, the diverse sparsity patterns that arise when mapping different stencil kernels onto TCUs obstruct awareness of and effectively address the distribution of sparsity within matrices, causing existing methods to overlook such diverse sparsity.

Secondly, \textit{how can we convert sparsity into the fine-grained structured patterns required by sparse TCUs?}   Sparse TCUs demand element-level alignment: every row must contain exactly two non-zeros within each four-element group. Yet stencil sparsity emerges from coarse-grained, spatially entangled structures that defy such regularity. This mismatch in granularity makes element-wise adjustment infeasible—nonzeros cannot simply be shuffled without corrupting stencil semantics. At the same time, coarse layout modifications—such as reordering rows or columns—often introduce cascading disruptions, where changes in one region break alignments elsewhere. As a result, even small transformations risk destabilizing the entire sparsity layout.

% Secondly, \textit{how can we convert the clustered sparsity into the structured sparsity required by sparse TCUs via PIT?} Although PIT is capable of redistributing sparse elements within matrices through row-level and column-level rearrangements, it operates at a coarse granularity, which contrasts sharply with the fine-grained, element-level constraints imposed by the 2:4 sparsity pattern. This disparity makes it exceedingly difficult to meet the strict requirements of the 2:4 pattern using such high-level transformation. Moreover, attempts to align the matrix with element-level constraints by PIT often lead to cascading disruptions throughout the sparsity structure, where modifications in one region propagate unpredictably across the matrix.

Thirdly, \textit{how can we automatically generate highly optimized CUDA kernels with optimal layout mapping?}  
Stencil kernels admit a vast space of layout-to-execution mappings, each inducing distinct sparsity patterns, memory traffic, and alignment constraints. This space is too large and irregular for manual tuning, forcing existing TCU-based methods to rely on hand-tuned schedules tailored to specific kernel sizes. Moreover, the intricate interplay between sparsity patterns, TCU constraints, and kernel performance metrics also presents a significant challenge in defining the search space for identifying the optimal layouts that maximize TCU utilization.

% Stencil workloads exhibit highly diverse kernel sizes and irregular sparsity patterns, making it difficult to define a unified mapping onto fixed-size sparse TCU fragments. Worse still, layout choices not only affect sparsity structure, but also reshape memory access patterns—introducing complex trade-offs between compute density and memory efficiency. Each candidate layout must be analyzed for 2:4 compatibility under structural constraints, while simultaneously evaluating its impact on memory traffic and scheduling. These intertwined dependencies defy decomposition, rendering manual tuning intractable and traditional compiler heuristics ineffective.

% Thirdly, \textit{how can we automatically generate highly optimized CUDA kernels accommodating specific sparse TCU architectures?} the diversity of stencil kernels and the dynamic characteristics of sparsity introduce substantial uncertainty, making it extremely difficult to develop automated methods that can generalize across workloads and hardware configurations, which forces all existing TCU-based methods to rely exclusively on case-by-case manual tuning. Moreover, the intricate interplay between sparsity patterns, TCU constraints, and kernel performance metrics also presents a significant challenge in defining the search space for identifying optimal matrix layouts that maximize sparse TCU utilization.

% Unstructured sparsity in TCU fragment across different stencil kernels for ConvStencil

\section{SparStencil}

In this section, we will provide a detailed introduction to SparStencil. It is a novel approach that leverages sparse TCUs for stencil computations for the first time, extending the boundaries of sparse TCUs for stencil computation.

% \subsection{Overview}

% We begin by providing an overview of the SparStencil. First, we introduce the \textit{Adaptive Layout Morphing} that transforms a wide range of stencil computations into hardware-aligned matrix multiplication while refines unstructured sparsity into a more structured, staircase-like sparsity pattern. Following this, we apply the \textit{Structured Sparsity Conversion} to orchestrate the staircase-like sparsity into a structured 2:4 sparsity via PIT. Finally, we employ \textit{Automatic Kernel Generation} dynamically to determine the optimal matrix layout based on the stencil input and the TCU fragment and integrate performance-boosting optimizations, achieving optimal performance.

\subsection{Adaptive Layout Morphing}

\label{morph}

Adaptive Layout Morphing is designed with two fundamental objectives: (1) to reformulate stencil computations into matrix representations flexibly, and (2) to implement a sparsity-aware strategy that refines clustered sparsity into a more structured, staircase-like sparsity pattern while consolidating consecutive redundant elements. These objectives are achieved through \textit{Stencil Flattening} and \textit{Duplications Crush}, respectively.

% \textit{Stencil Flattening} is a generalized approach that transforms TCU-incompatible stencil computations with varying kernel sizes into a unified vector-matrix multiplication format that is compatible with TCUs. The process begins by flattening the values of the stencil kernel into a single row vector, termed the kernel vector. Simultaneously, the region of the input data that the stencil kernel covers is flattened into corresponding column vectors. These column vectors, derived from each position the stencil kernel moves over, are then concatenated row-wise to form the input matrix. By multiplying the kernel vector with this constructed input matrix, the stencil computation—originally performed as a series of weighted summations—can be represented and executed as vector-matrix multiplication. Figure \ref{fig:flattening} illustrates the stencil flattening process for a $3 \times 3$ kernel applied to a $5 \times 5$ input matrix.

\textit{Stencil Flattening} generalizes the transformation of stencil computations—originally incompatible with TCU—into a unified vector-matrix multiplication format suitable for TCUs. Specifically, the approach unfolds the stencil kernel weights into a single-row vector, the \textit{kernel vector}, while simultaneously reshaping each region of the input data covered by the sliding kernel into corresponding column vectors. Concatenating these column vectors row-wise forms the \textit{input matrix}. Consequently, stencil computations, previously executed as repeated weighted summations, become a vector-matrix multiplication. Figure \ref{fig:flattening} illustrates the stencil flattening process for a $3 \times 3$ kernel applied to a $5 \times 5$ input matrix.

\begin{figure}[htbp]
\centering
\includegraphics[width=0.5\textwidth]{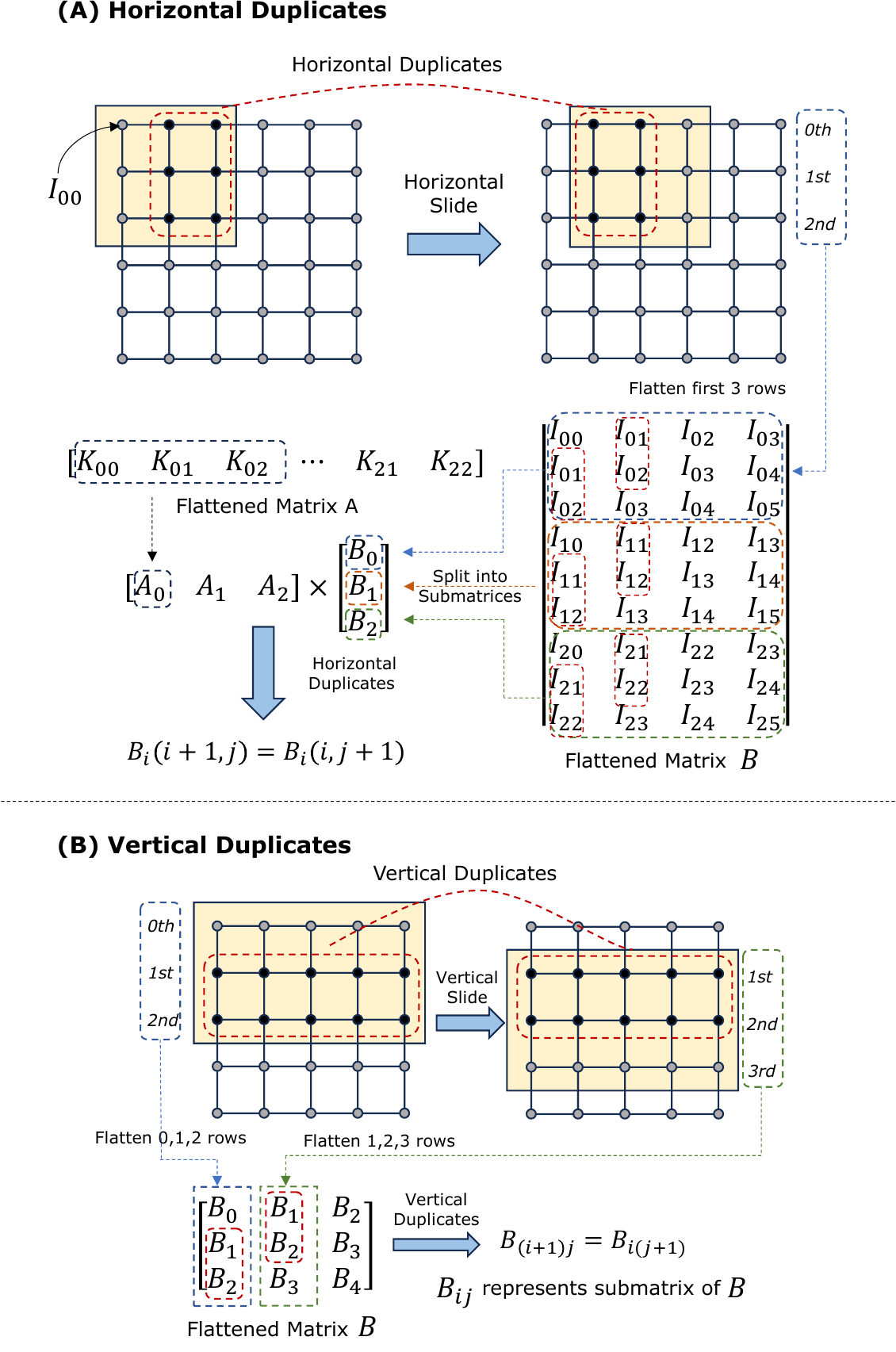}
\caption{The distribution of duplicates after stencil flattening. (a) The horizontal duplicate distribution at the element level within a sub-matrix. (b) The vertical duplicates distribution between sub-matrices at the sub-matrix level.}
\label{fig:horizontal}
\end{figure}

However, directly applying stencil flattening to TCUs often leads to low hardware utilization, as the kernel vector occupies only one row of the TCU fragment, leaving other rows unused. Moreover, this transformation inherently generates significant redundancy: the input matrix contains numerous duplicated elements, aggravating memory-bound constraints inherent in stencil computations.

% More generally, by flattening the kernel weights of any size into a row vector and the input elements covered by the sliding kernel into column vectors, Stencil Flattening transforms stencil computations---which were originally unsuitable for execution on TCU---into a unified vector-matrix multiplication format that is compatible with TCUs. However, despite the support of TCU, the vector-matrix multiplication often results in extremely low TCU utilization because the vector only maps to a single row of the TCU fragment, leaving other rows unused. Additionally, the matrix obtained after stencil flattening contains a large number of duplicates. Such a massive, repetitive matrix exacerbates the memory-bound issues inherent to stencil computations.

To address these issues, we introduce \textit{Duplicates Crush}.\footnote{The inspiration for Duplicates Crush comes from the game Candy Crush, where elements are eliminated in a similar manner. Here, we analogously eliminate duplicates.} By crushing redundant elements within the input matrix, it adaptively transforms the \textit{vector-matrix} multiplication into \textit{matrix-matrix} multiplication according to the TCU fragment size while intentionally guiding a staircase-like sparsity pattern, laying the groundwork for the subsequent transformation. The \textit{Duplicates} refers to the element duplicates caused by the horizontal and vertical sliding of the kernel in stencil computations, which can be categorized along the two directions into two types: horizontal duplicates and vertical duplicates.

The horizontal duplicates are illustrated in Figure \ref{fig:horizontal}(a). During horizontal sliding, the stencil kernel repeatedly covers certain elements due to overlapping regions at adjacent positions. These repeated elements are reflected in matrix $B'$ after Stencil Flattening, such as $I_{01}$ and $I_{02}$, $I_{11}$ and $I_{12}$, as well as $I_{21}$ and $I_{22}$. In Figure \ref{fig:horizontal}(a), $B'$ is partitioned into three submatrices, where each $B_i$ consists solely of elements from the $i$-th input row. Due to horizontal kernel sliding, adjacent kernels in Figure \ref{fig:horizontal}(a) share two elements per row. This overlap induces horizontal duplicates in $B$, with adjacent columns in each $B_i$ sharing two identical elements. Generally, these horizontal duplicates can be mathematically expressed as:
\begin{equation} 
B_i(i+1, j) = B_i(i, j+1)
\label{horizontal deduplicates}
\end{equation}
where $B_i \in \mathbb{R}^{k \times (n-k+1)}$, $B \in \mathbb{R}^{k^2 \times (m-k+1)(n-k+1)}$  and $B_i$ is a submatrix of $B$. 

% The horizontal duplicates are illustrated in Figure \ref{fig:horizontal}(a). During horizontal sliding, the stencil kernel repeatedly covers certain elements due to overlapping regions at adjacent positions. These repeated elements are reflected in matrix $B'$ after Stencil Flattening, such as $I_{01}$ and $I_{02}$, $I_{11}$ and $I_{12}$, as well as $I_{21}$ and $I_{22}$. In Figure \ref{fig:horizontal}(a), we divide the matrix into three submatrices, where the elements in $B_i$ are entirely composed of the $i$-th row of the input. The horizontal duplicates arise due to the horizontal sliding of the kernel, so in Figure \ref{fig:horizontal}(a), two adjacent horizontal kernels will simultaneously cover two elements in each row. This is reflected in the flattened matrix $B$, where each submatrix $B_i$ has two identical elements in adjacent columns. Generally, These horizontal duplicates can be mathematically expressed as:

% \begin{figure}[htbp]
% \includegraphics[width=0.5\textwidth]{./image/PIT.png} % Adjust width here
% \caption{An Example of Permutation Invariant Transformation. The figure shows swapping $1$-th and $2$-th columns of matrix $ A $, while simultaneously swapping the $1$-th and $2$-th rows of matrix $ B $. According to the commutative property of addition, this ensures that the result of the matrix multiplication remains unchanged.}
% \label{PIT}
% \end{figure}

\begin{figure}[t]
\centering
\includegraphics[width=0.5\textwidth]{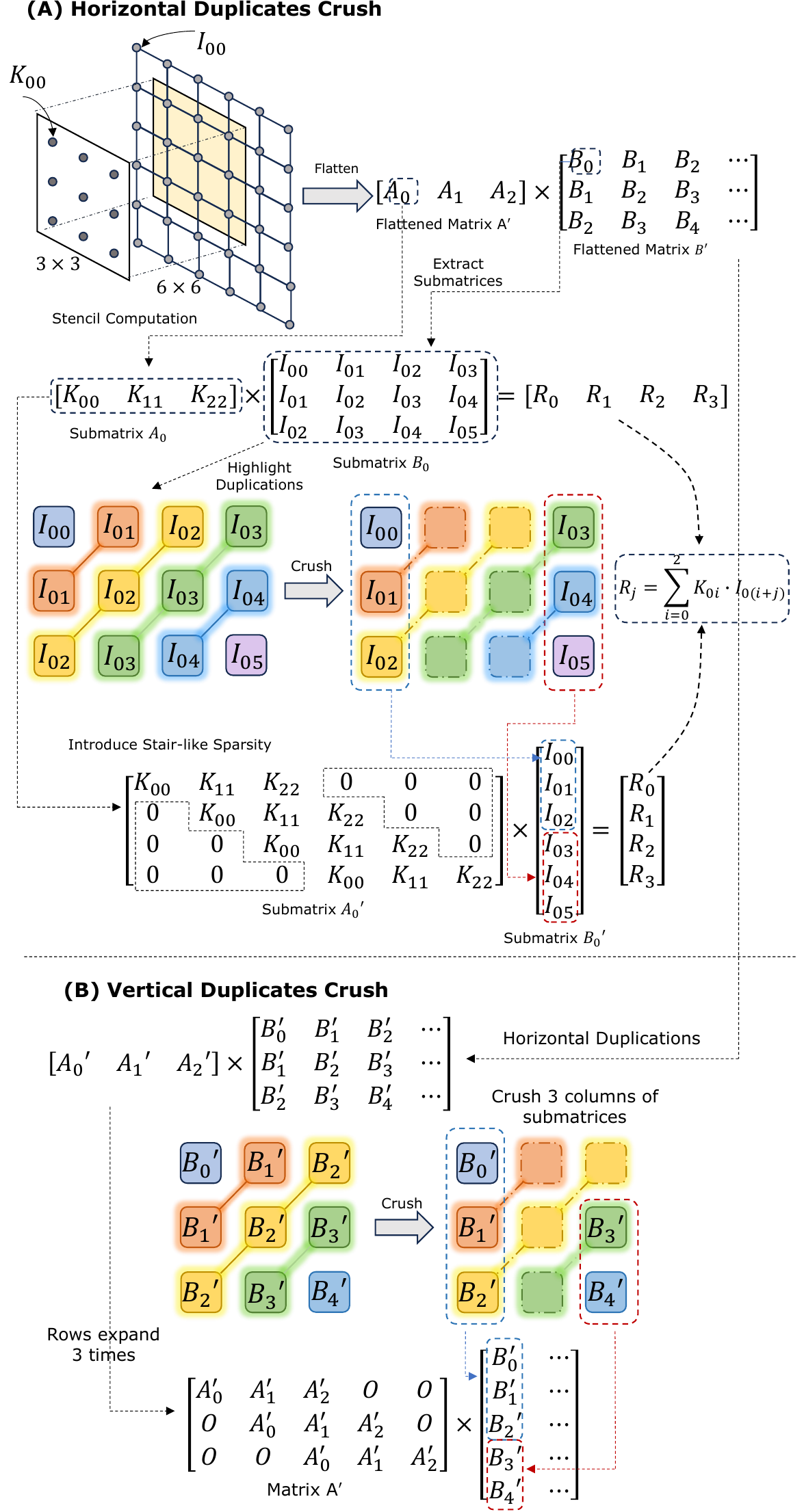} % Adjust width here
\caption{The process of Duplicates Crush. Figure (a) illustrates the horizontal duplicates crush at the element level within a sub-matrix. Figure (b) shows the vertical duplicates crush between sub-matrices at the sub-matrix level.}
\label{fig:crush}
\end{figure}
The vertical duplicates are caused by the vertical sliding of the stencil kernel. As illustrated in Figure \ref{fig:horizontal}(b), elements from the first row and the second row are covered simultaneously by these adjacent kernels. After flattening, these repeated elements appear as submatrices in $B'$, where the elements from the first row become submatrix $B_1$ and the elements from the second row become submatrix $B_2$. Viewing the matrix $B'$ as consisting of submatrix $B_i$, each composed of the elements from the $i$-th row, the vertical redundancy in $B'$ manifests as duplication between these submatrices. Mathematically, this can be expressed as:

\begin{equation} 
B’_{i+1,j} = B'_{i,j+1}
\label{vertical deduplicates}
\end{equation}
where $B'_{i,j}$ denotes the element located at the $(i, j)$ position within the submatrix of $B'$.

% \begin{figure*}[htbp]
% \includegraphics[width=1\textwidth]{./image/pit.pdf} % Adjust width here
% \caption{The Process of Structured Sparsity Conversion through PIT.}
% \label{fine-tuning}
% \end{figure*}

After analyzing the distribution of duplicates, we now introduce the \textit{Duplicates Crush}. It eliminates the duplicates in matrix $B'$ after \textit{Stencil Flattening}, while expanding the vector $A$ into a matrix $A'$ and intentionally introducing a staircase-like sparsity pattern within $A$ to prepare for the transformation. Similar to how duplicates are distributed, the crushing procedure is categorized along the stencil kernel’s two sliding directions: horizontal and vertical.

% After introducing the distribution of duplicates, we now present how to effectively crush these duplicates. The \textit{Duplicates Crush} process eliminates the horizontal and vertical duplicates within the large matrix $B'$ in the vector-matrix multiplication after \textit{Stencil Flattening}, while expanding the vector $A$ into a matrix $A'$ and intentionally introducing a staircase-like sparsity pattern within $A$ to prepare for the PIT. Similar to the distribution of duplicates, the crushing of duplicates can also be categorized along the two directions of the stencil kernel’s sliding into two types: horizontal duplicate crushing and vertical duplicate crushing.

We begin with horizontal duplicate crushing, as characterized in Equation~\ref{horizontal deduplicates}. Since these duplicates are confined within each sub-matrix $B_i$ of $B'$, it suffices to illustrate the process at the sub-matrix level for uniform application across all $B_i$. The procedure, illustrated in Figure~\ref{fig:crush}(a), proceeds as follows:
(1) Consider the flattened matrices $A'$ and $B'$ in sub-matrix form, and select $A_0$ and $B_0$ as representative examples.
(2) In $B_0$, duplicate values appear diagonally across columns. These are eliminated by retaining one representative per group and reorganizing the unique entries into a single-column format.
(3) Concurrently, $A_0$ is vertically expanded into multiple rows with deliberate zero-padding to align with the new structure, forming a staircase-like sparsity pattern. Although not directly compatible with sparse TCUs, the staircase pattern imposes a regular layout that facilitates the subsequent transformation.

We next describe vertical duplicate crushing, which targets duplicates across sub-matrices in $B'$, as defined in Equation~\ref{vertical deduplicates}. The crushing process is again applied at the sub-matrix level, as illustrated in Figure~\ref{fig:crush}(b), and proceeds as follows: 
(1) After the horizontal duplicate crushing, vertical duplicates in $B'$ exhibit diagonal alignment at the sub-matrix level, similar to the horizontal case. (2) Redundant sub-matrices across the three columns are removed, retaining one representative per group. The remaining sub-matrices are then reorganized into a single column structure. (3) Meanwhile, $A$ is similarly expanded into multiple rows at the sub-matrix level.

The above procedures illustrate both horizontal and vertical duplicate crushing. In general, horizontal crushing can be applied to every $r_1$ columns within each $B_i$, while vertical crushing merges every $r_2$ columns of $B'$. The values of $r_1$ and $r_2$ are determined during the subsequent layout exploration phase.

\subsection{Structured Sparsity Conversion}
\label{sparse_conversion}

% Despite leveraging \textit{Adaptive Layout Morphing} to intentionally guide a staircase-like sparsity pattern, the resulting sparsity remains clustered and does not satisfy the 2:4 sparsity requirement of sparse TCUs. To address this issue, we introduce \textit{Structured Sparsity Conversion} to fine-tune the matrix layout, reorganizing the staircase-like sparsity into a structured format that conforms to the 2:4 sparsity requirement.

Although \textit{Adaptive Layout Morphing} guides sparsity into a staircase-like structure, the resulting pattern remains irregular and violates the 2:4 constraint required by sparse TCUs. To bridge this gap, we propose \textit{Structured Sparsity Conversion}, a general transformation that reorganizes intermediate sparsity into a hardware-aligned 2:4 format.

A core component of this process is the \textit{Permutation Invariant Transformation (PIT)}, which preserves computation semantics during layout reordering. PIT simultaneously permutes the columns of $\mathbf{A}$ and the corresponding rows of $\mathbf{B}$ along the $k$-dimension, ensuring the result of $\mathbf{A} \times \mathbf{B}$ remains unchanged. Formally, let $\mathbf{A} = [\mathbf{a}_1, \dots, \mathbf{a}_k] \in \mathbb{R}^{m \times k}$ and $\mathbf{B} = [\mathbf{b}_1^\top; \dots; \mathbf{b}_k^\top] \in \mathbb{R}^{k \times n}$. Then,
\begin{equation}
\mathbf{C} = \sum_{i=1}^k \mathbf{a}_i \mathbf{b}_i^\top = \sum_{i=1}^k \mathbf{a}_{\mathcal{P}(i)} \mathbf{b}_{\mathcal{P}(i)}^\top,
\end{equation}
guaranteeing that the multiplication remains invariant under any shared permutation $\mathcal{P}$—a crucial property for enabling valid sparse pattern adaptation.

While PIT guarantees correctness, not all permutations yield valid 2:4 structured sparsity. To make the search space tractable, we observe that the 2:4 constraint can be equivalently expressed as a composition of 1:2 and 0:2 subpatterns (Section~\ref{se:TCU}). This motivates the core problem of \textit{Structured Sparsity Conversion}: finding a permutation $\mathcal{P}$ that rearranges nonzeros into conflict-free 1:2 or 0:2 groupings while minimizing padding. We formalize this as follows:

% Building upon PIT, \textit{Structured Sparsity Conversion} focuses on finding the optimal permutation $\mathcal{P}$. To this end, we decompose the 2:4 sparsity constraint into a set of 1:2 and 0:2 patterns, which are inherently compatible with sparse TCU execution (Section~\ref{se:TCU}). This leads to the following formal formulation:

Given a matrix \( \mathbf{A'} \in \mathbb{R}^{m \times n} \) derived from layout morphing with parameters \( (r_1, r_2) \), we construct an augmented matrix \( \mathbf{A'_{pad}} \in \mathbb{R}^{m \times (n + p)} \) by appending \( p \) zero columns. The goal is to find a permutation \( \mathcal{P} \) such that the transformed matrix \( \mathbf{A''} = \mathcal{P}(\mathbf{A'}) \) satisfies a 1:2 or 0:2 sparsity pattern while minimizing \( p \). We model this as a graph-theoretic matching problem:
\begin{definition}[Conflict Graph]
\label{def:conflict}
Given \( \mathbf{A'} \in \mathbb{R}^{m \times n} \), its \emph{conflict graph} is \( G = (V, E) \), where \( V = \{v_i\}_{i=1}^n \) represents the columns of \( \mathbf{A'} \), and \( (v_i, v_j) \in E \) if and only if \( \exists r \in [m] \) such that \( \mathbf{A'}[r, i] \neq 0 \) and \( \mathbf{A'}[r, j] \neq 0 \).
\end{definition}
\begin{definition}[Augmented Matching Graph]
Given zero columns \( Z = \{z_1, \dots, z_p\} \), the \emph{augmented graph} is \( G' = (V \cup Z, E) \), where \( E \) is inherited from the original conflict graph and no additional edges are introduced.
\end{definition}
\begin{definition}[Valid Matching]
\label{def:matching}
A matching \( \mathcal{M} \subset \binom{V \cup Z}{2} \) is valid if it satisfies: (i) \emph{coverage}: each \( v \in V \) appears in exactly one edge in \( \mathcal{M} \); and (ii) \emph{conflict-freedom}: \( \mathcal{M} \cap E = \emptyset \).
\end{definition}

\begin{problem}[Minimum Zero-Column Matching]
\label{prob:minzero}
Find the minimal \( p \in \mathbb{N} \) such that a valid matching \( \mathcal{M} \subset \binom{V \cup Z}{2} \) exists in \( G' = (V \cup Z, E) \), with \( |Z| = p \).
\end{problem}

%突然有个灵感，可以画两个平面，然后两个平面对映两个层次，在两个屏幕之间做映射，更加清楚
\begin{figure}[htbp]
\includegraphics[width=0.5\textwidth]{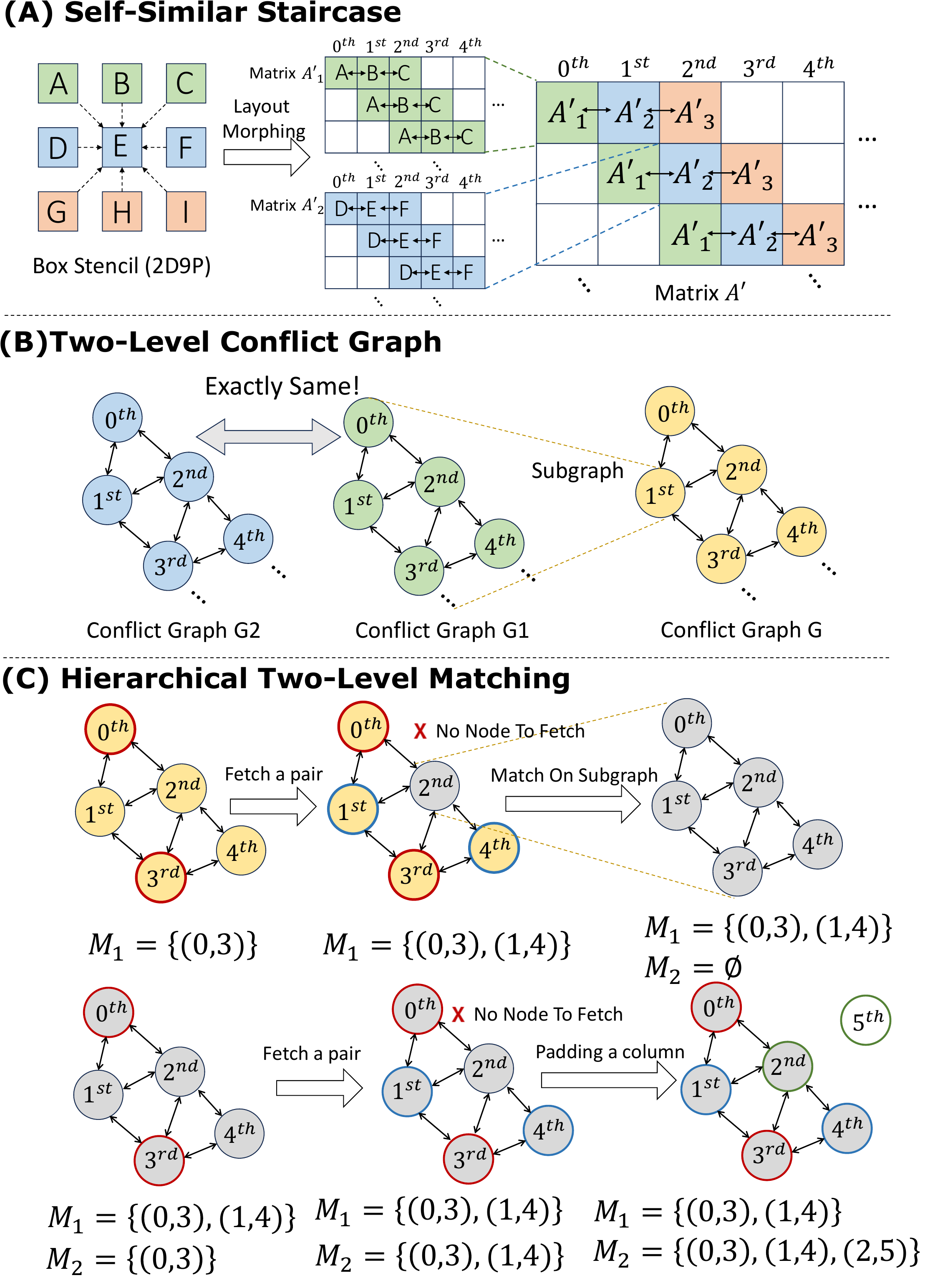} % Adjust width here
\caption{An example of Sparsity Conversion. (a) The self-similar staircase property in the matrix undergoing layout morphing. (b) The global conflict graph and the local conflict graph. (c) The process of hierarchical two-level matching algorithm.}
\label{fig:conflict}
\end{figure}
To solve the problem, we introduce the following definition:
\begin{definition}[$k$-Staircase Matrix]
A matrix \(\mathbf{A'} \in \mathbb{R}^{m \times n}\) has the \emph{$k$-staircase property} if:  
\[
\mathbf{A'}[r,c] \neq 0 \iff r \leq c < r + k -1 \quad (\forall r,c \in [n])
\]  
where \([n] := \{0,1,...,n-1\}\) denotes indices.
\end{definition}

As described in Section~\ref{morph}, the matrix \( \mathbf{A'} \) can be partitioned into blocks \( \{\mathbf{A'}_{i,j}\} \), exhibiting a self-similar $k$-staircase structure:  
(i)~\emph{Global Staircase}: treating each block \( \mathbf{A'}_{i,j} \) as a scalar, the block-level matrix satisfies the $k$-staircase property;  
(ii)~\emph{Local Staircase}: each nonzero block \( \mathbf{A'}_{i,j} \) individually satisfies the same property.  
Figure~\ref{fig:conflict} (a) illustrates this self-similarity.

The self-similar structure of \( \mathbf{A'} \) induces two levels of conflict graphs, as illustrated in Figure~\ref{fig:conflict} (b):  
(i)~\emph{Global Conflict Graph} \( G = (V, E) \): each node corresponds to a block column \( \mathbf{A'}_{:,i} \), with \( (i, j) \in E \) iff \( \exists\, r \) such that \( \mathbf{A'}_{r,i} \neq \mathbf{0} \) and \( \mathbf{A'}_{r,j} \neq \mathbf{0} \);  
(ii)~\emph{Local Conflict Graph} \( G_i = (V_i, E_i) \): each node corresponds to a column in \( \mathbf{B'}_i \), with \( (p, q) \in E_i \) iff \( \exists\, r \) such that \( \mathbf{A'}_{i,j}[r,p] \neq 0 \) and \( \mathbf{A'}_{i,j}[r,q] \neq 0 \).

After constructing two-level conflict graphs, the following theorem can be derived based on the $k$-staircase property:

% \textbf{Theorem 2 (Structural Properties of Staircase Conflict Graphs)}
\begin{theorem}[Non-Conflict Property]
\label{th:strucuture}
Let \(\mathbf{A'}\) be a matrix with a staircase pattern and its conflict graph \(G = (V, E)\). Then any two nodes separated by at least $k$ hops are \emph{conflict-free}:
    \[
    \forall j \in [n - \delta], \quad (c_j, c_{j+\delta}) \notin E \quad \text{for all } \delta \geq k
    \]
\end{theorem}
\begin{proof}
\textit{\emph{Conflict free}}: By the staircase offset, columns \(c_j\) and \(c_{j+\delta}\) satisfy:  
\[
\forall r, \mathbf{A'}[r,c_j] \neq 0 \implies \mathbf{A'}[r,c_{j+\delta}] = 0
\]  
Thus, no shared nonzero rows exist. \qedhere  
\end{proof}

With Theorem \ref{th:strucuture}, we propose \textit{Hierarchical Two-Level Matching} to solve the Problem \ref{prob:minzero}, which is shown in Algorithm \ref{alg:hierarchical_matching}. The process of Algorithm \ref{alg:hierarchical_matching} is shown in Figure \ref{fig:conflict}(c). Furthermore, Theorem \ref{th:solution} establishes the correctness and optimality of this algorithm.
\begin{theorem}[Correctness and Optimality]
\label{th:solution}
Let \( \mathcal{M} \) be the matching returned by Algorithm~\ref{alg:hierarchical_matching}, and let \( Z \) be the set of inserted zero-columns. Then:  
(i)~\emph{Validity}: for all \( (v_i, v_j) \in \mathcal{M} \), we have \( |j - i| \ge k \);  
(ii)~\emph{Minimality}: \( |Z| = p \) achieves the minimum value in Problem~\ref{prob:minzero}.
\end{theorem}

\begin{proof}
(i)~\emph{Validity.} Let \( (u, v) \in \mathcal{M} \). If \( (G_i, G_j) \in \mathcal{M}_1 \), then \( j - i = s_1 \ge k \); if \( (v_i, v_j) \in \mathcal{M}_2 \), then \( j - i = s_2 \ge k \). Hence, all matched pairs satisfy \( |j - i| \ge k \). (ii)~\emph{Minimality.} Since all subgraphs are structurally identical, it suffices to analyze a single subgraph \( G_i \) of size \( g \). If \( k \le \lfloor g/2 \rfloor \), then: when \( g \) is even, a perfect matching exists; when odd, exactly one node remains unmatched—this is trivially optimal.  
If \( k > \lfloor g/2 \rfloor \), the algorithm pairs \( v_1,\dots,v_{g-k} \) with \( v_{k+1},\dots,v_g \), matching \( g - k \) nodes. Assume a better matching \( \mathcal{M}' \) exists. Validity requires \( j \ge i + k \Rightarrow i \le g - k \) for any pair \( (i,j) \in \mathcal{M}' \), so no pair with \( i > g - k \) can exist, contradicting the assumption. Thus, the matching is optimal. \qedhere
\end{proof}

Theorem~\ref{th:solution} ensures that Algorithm~\ref{alg:hierarchical_matching} achieves both correctness and optimality with linear complexity \( O(|V|) \), since each vertex is traversed at most once. While this is guaranteed under the $k$-staircase assumption, our approach remains broadly applicable even when the input deviates from this structure. In such cases, the same conflict graph formulation allows us to fall back to the classical Blossom algorithm~\cite{Edmonds_1965} to compute maximum matchings over arbitrary sparsity patterns. Despite its worst-case complexity of \( O(|E||V|^2) \),  the Blossom algorithm remains practical in our setting, as real-world stencil kernels typically yield small and sparse conflict graphs.

\begin{algorithm}[t]
\caption{Hierarchical Two-Level Matching}
\label{alg:hierarchical_matching}
\KwIn{Global graph $G$ with $n$ nodes, subgraph size $g$, stencil size $k$}
\KwOut{Perfect matching $\mathcal{M}$}
Initialize $m \leftarrow n/g$, $s_1 \leftarrow \max(\lfloor m/2 \rfloor, k)$, $\mathcal{M}_1 \leftarrow \emptyset$\;
\For{$i = 0$ \KwTo $m - 1$}{
  \If{$G_i$ unmatched and $i + s_1 < m$}{
    $\mathcal{M}_1 \leftarrow \mathcal{M}_1 \cup \{(G_i, G_{i+s_1})\}$\;
  }
}
Initialize $\mathcal{M}_2 \leftarrow \emptyset$, $s_2 \leftarrow \max(\lfloor g/2 \rfloor, k)$\;
\For{each unmatched $G_x$}{
  \For{$u = 0$ \KwTo $|V(G_x)| - 1$}{
    $v \leftarrow u + s_2$\;
    \If{$v_x^u$ unmatched}{
      \uIf{$v < |V(G_x)|$}{
        $\mathcal{M}_2 \leftarrow \mathcal{M}_2 \cup \{(v_x^u, v_x^v)\}$\;
      }
      \Else{
        Add zero node $\zeta_v$, $\mathcal{M}_2 \leftarrow \mathcal{M}_2 \cup \{(v_x^u, \zeta_v)\}$\;
      }
    }
  }
}
Initialize $\mathcal{M} \leftarrow \emptyset$\;
\For{$(G_p, G_q) \in \mathcal{M}_1$, $t = 0$ \KwTo $g - 1$}{
  $\mathcal{M} \leftarrow \mathcal{M} \cup \{(v_p^t, v_q^t)\}$\;
}
$\mathcal{M} \leftarrow \mathcal{M} \cup \mathcal{M}_2$\;
\end{algorithm}

\subsection{Automatic Kernel Generation}
\label{se:auto}
% Even with methods for leveraging sparse TCUs, the presence of numerous different fragments and the constantly changing size of stencil kernels make manual, end-to-end CUDA kernel tuning for each specific size highly inefficient and lacking in scalability. To address this issue, we propose \textit{Automatic Kernel Generation}. It dynamically determines the optimal matrix layout based on the stencil input and the TCU fragment and integrates performance-boosting optimizations, automatically generating highly-optimized CUDA kernel often surpassing hand-tuned expert implementations.

% Need to write that "handcraft cannot achieve"

% Despite methods for leveraging sparse TCUs, the presence of numerous TCU fragments and the dynamically varying sizes of stencil kernels make manual, end-to-end CUDA kernel tuning for each specific size highly inefficient and lacking in scalability. To address this issue, we propose \textit{Automatic Kernel Generation}, a compiler framework designed with three phases: (1) \textit{layout exploration} to dynamically determine the optimal data layout tailored to specific stencil patterns and TCU fragments, minimizing memory transfer overhead while maximizing computational density. (2) \textit{sparsity conversion} to automatically apply the most efficient PIT to organize clustered sparsity, (3) and \textit{code generation} to analyze the metadata and the lookup table for data layout and integrate performance-boosting optimizations to generate highly-optimized CUDA kernel surpassing hand-tuned expert implementations. 

Despite methods for leveraging sparse TCUs, manual tuning for sparse TCUs is both labor-intensive and inherently suboptimal due to the high variability of stencil kernels and TCU fragments. To address this issue, we propose \textit{Automatic Kernel Generation}, a compiler framework that enables end-to-end scheduling through two phases: (1) \textit{layout exploration} dynamically determines the optimal data layout by automatically searching over a tunable space, minimizing memory transfer overhead while maximizing computational density.  (2) \textit{Code Generation} emits expert-level CUDA kernels via metadata-driven sparse MMA, lookup-based memory mapping, and double-buffered pipelines, consistently outperforming hand-tuned implementations.

% (2) \textit{sparsity conversion} reformulates sparsity alignment as a hierarchical graph matching problem, with a provably optimal two-level algorithm that converts irregular sparsity into TCU-aligned patterns with minimal padding.

\textbf{\textit{Layout Exploration}}. The configuration parameters $(r_1, r_2)$ in Section \ref{morph} modulate data layout patterns and induce clustered sparsity in TCU, where $r_1$ and $r_2$ denote the number of consecutive columns to be merged in $B'$ and $B$. The configuration creates distinct memory-computation trade-offs that necessitate automated exploration to identify optimal configurations. To systematically evaluate these configurations, we adopt an analytical model tailored  for stencil computations on TCU \cite{10.1145/3627535.3638476}, shown in Equation \ref{eq:T}, \ref{eq:com} and \ref{eq:mem}.
\begin{equation}\label{eq:T}
    T=max(T_{compute}\;, T_{memory}) 
\end{equation}
\begin{equation}\label{eq:com}
T_{compute}=\frac{ N_{MMA} \times CPI_{tcu}}{fN_{tcu}}
\end{equation}
\begin{equation}\label{eq:mem}
T_{memory} = max(\frac{data_{R}}{bw_G} + \frac{data_{W}}{bw_G}, \frac{data_{transW}}{bw_{S}} + \frac{data_{transR}}{bw_{S}})
\end{equation}

\textit{Compute Evaluation}. As defined in Equation~\ref{eq:com}, $T_{compute}$ depends on $N_{MMA}$, with $CPI_{tcu}$, $f$, and $N_{tcu}$ fixed by hardware. Given an $m \times n$ input, a $k \times k$ stencil kernel, and a TCU fragment of size $M \times K \times N$, applying $(r_1, r_2)$ layout morphing yields $\mathbf{A'} \in \mathbb{R}^{m' \times k'}$, $\mathbf{B'} \in \mathbb{R}^{k' \times n'}$ with $m' = r_1 r_2$, $k' = (k + r_1 -1)(k + r_2 -1)$, and $n' = \frac{(m-k+1)(n-k+1)}{r_1 r_2}$. The corresponding MMA count is:
\begin{equation}
 \label{eq:compute2}
    N_{MMA} = \left\lceil \frac{m'}{M} \right\rceil \left\lceil \frac{k'}{K} \right\rceil \left\lceil \frac{n'}{N} \right\rceil
\end{equation}
Substituting $N_{MMA}$ into Equation~\ref{eq:com} yields $T_{compute}$.
\begin{table}[b]
  \centering
  \caption{Summary of symbols used in performance modeling.}
  \label{tab:notation}
  \small
  \begin{tabular}{cc}
    \toprule
    \textbf{Symbol} & \textbf{Description} \\
    \midrule
    $T$, $T_{\text{compute}}$, $T_{\text{memory}}$ & Total, compute, and memory time \\
    $data_{\text{transR}}$ / $data_{\text{transW}}$ & Shared memory read/write volume \\
    $N_{\text{MMA}}$, $N_{\text{tcu}}$ & MMA operations number, TCUs number\\
    $data_{R}$ / $data_{W}$ &Global memory read/write volume\\
    $CPI_{\text{tcu}}$, $f$ & Cycles per TCU op, GPU core frequency \\
    $bw_G$, $bw_S$ & Bandwidth of global/shared memory \\
    \bottomrule
  \end{tabular}
\end{table}

\textit{Memory Evaluation}. SparStencil stores the original input $\mathbf{I}$ and the pre-converted kernel matrix $\mathbf{A'}$ in global memory. As stencil kernels are typically small, the footprint of $\mathbf{A'}$ is negligible compared to $\mathbf{I}$, which remains layout-invariant and thus independent of $(r_1, r_2)$. Consequently, $data_R$ and $data_W$ in Equation~\ref{eq:mem} are constant, while $data_{transR}$ and $data_{transW}$ depend on $\mathbf{A'}$ and $\mathbf{B'}$ after layout morphing:
\begin{equation}
    data_{transR} = data_{transW} = k'\left(\frac{m}{2} + n'\right)
\end{equation}
We compute $T_{\mathrm{memory}}$ via Equation~\ref{eq:mem}, and total latency $T$ using Equation~\ref{eq:T}. The optimal configuration is selected by exhaustive search:
\begin{equation}
(r_1^*, r_2^*) = \arg\min_{r_1, r_2 \in \mathcal{S}} T(r_1, r_2)
\end{equation}

For each layout candidate $(r_1, r_2)$, SparStencil conducts static graph analysis to construct a conflict graph from the symbolic sparsity of $\mathbf{A'}$. A provably optimal permutation is then derived via Algorithm~\ref{alg:hierarchical_matching} to satisfy the 2:4 sparsity constraint under structural constraints. The resulting permutation is applied via PIT, yielding a layout conforming to sparse TCU execution.

 % \textbf{\textit{Sparsity Conversion}}. Following layout optimization, the determined parameters $(r_1^*, r_2^*)$ prescribe an optimal data organization. Nevertheless, this introduces a critical implementation challenge: our proposed PIT initially assumes $r_1 = k+1$ and $r_2 = \lfloor M/m \rfloor$ for structure sparse pattern (Section~\ref{sparse_conversion}).  Yet, after Layout Exploration, these values may change, necessitating a generalized sparsity conversion mechanism to accommodate arbitrary configuration parameters.

% To address this challenge, we decompose the standard 2:4 sparsity constraint into combinations of 1:2 or 0:2 patterns. While this approach may introduce sub-2:4 patterns (e.g., 0:4 or 1:4), such configurations remain compatible with existing processing frameworks (Section~\ref{se:TCU}). For matrices that inherently defy 2:4 grouping (e.g., those with odd column counts), unmatched columns are padded with all-zero columns to enforce 1:2 or 0:2 sparsity. However, padding all-zero columns incurs significant sparsity overhead. This motivates the need to identify an optimal PIT configuration that minimizes redundant zero-column insertion while strictly adhering to sparsity constraints. Mathematically, the problem can be formally represented as:

% For matrices that inherently defy 2:4 grouping (e.g., those with odd column counts), unmatched columns are padded with all-zero columns to enforce 1:2 or 0:2 sparsity. Since padding incurs redundancy, we seek an optimal PIT configuration that minimizes zero insertions while preserving structural correctness. 

\textbf{\textit{Code Generation}}. With optimal parameters ($r_1$, $r_2$) and the PIT transformation applied, the resulting irregular stencil patterns and hardware-specific layouts create highly scattered memory access patterns that are impractical to analyze manually. To address this, SparStencil automates CUDA kernel generation by integrating metadata encoding, lookup table acceleration, and a highly optimized CUDA template.

\begin{figure*}[htbp]
\includegraphics[width=0.95\textwidth]{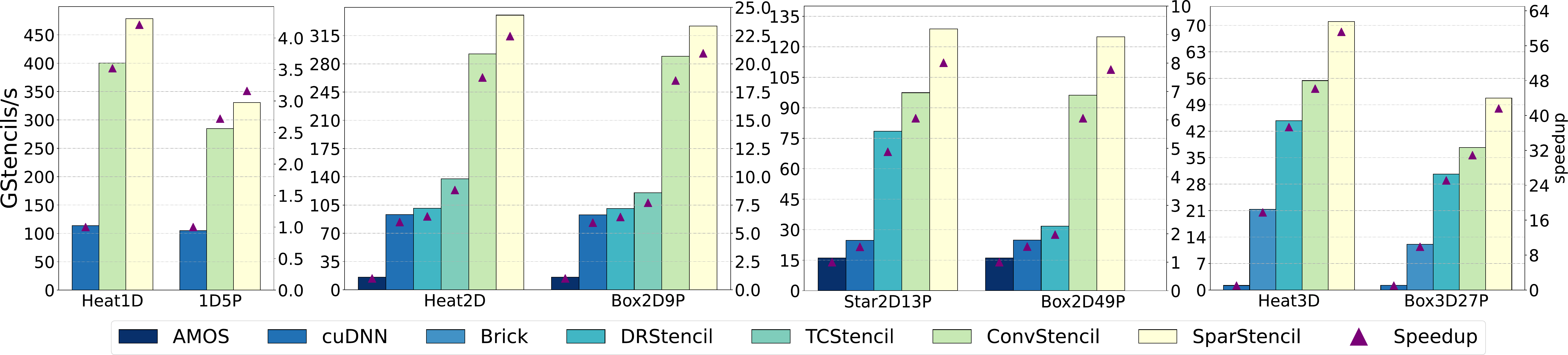} % Adjust width here
\caption{Performance comparison of SparStencil with state-of-the-arts.}
\label{sota}
\end{figure*}
% \textbf{\textit{Code Generation}}. At this stage, we have determined  the optimal parameters ($r_1$, $r_2$) and applied the PIT transformation. However, the irregular shapes of stencil patterns and the complexity of hardware-specific transformations result in highly scattered data layouts that are impractical to manually analyze. To address this, in this phase, SparStencil automates CUDA kernel generation by integrating metadata encoding, lookup table acceleration and a highly-tuned  CUDA template.

Metadata. Computing sparse matrix multiplication $\mathbf{A} \times \mathbf{B}$ on sparse TCUs necessitates the generation of metadata to indicate the positions of non-zero elements in matrix $\mathbf{A}$. The metadata generation process is straightforward: it involves traversing matrix $\mathbf{A}$,  recording the indices of nonzero elements, and encoding them into a compact format. 

Lookup Table. During layout transformation, computing address pointer offsets for transferring data from global to shared memory introduces significant computational overhead due to frequent integer division and modulus operations, which are inefficient on GPUs. Moreover, these computations often exhibit redundancy across different processing blocks. To mitigate this, we precompute pointer offsets on the host and store them in lookup tables, which are then passed to the CUDA kernel. This approach eliminates redundant computations, significantly improving memory access efficiency.

Based on prior analysis, SparStencil instantiates optimized CUDA kernels from a predefined template tailored for memory efficiency, compute-transfer overlap, and sparse TCU utilization. The kernel follows a three-stage pipeline: (1) asynchronous loading via lookup tables to prefetch global memory into shared memory; (2) sparse MMA execution using metadata, interleaved with the next load to maximize throughput; and (3) writing results back to global memory, completing a double-buffered loop.

\section{Evaluation}
\label{se:evaluation}
\subsection{Setup}

\textbf{\textit{Implementation.}} We implement SparStencil with CUDA C++ and PTX. SparStencil is compiled with NVCC 12.4. 

\textbf{\textit{Machine.}} Our platform is composed of an AMD EPYC 7543 processor and an Nvidia A100 GPU with sparse Tensor Core. The A100 GPU is connected to the motherboard via PCIe Gen4. The A100 GPU is equipped with 108 SMs, each containing 4 sparse Tensor Cores~\cite{nvidia_a100_datasheet}.

% \begin{table}[h]
%   \centering
%   \caption{The Machine Specifications}
%   \label{tab:configuration}
%   \small
%   \begin{tabular}{ccc}
%     \toprule
%     \textbf{Name} & \textbf{Confiuration} \\
%     \midrule
%         Server & AS-4124GS-TNR                     \\
%         CPU           & AMD EPYC 7543 32C/64T/2.8GHz  \\
%         Memory        & 64GB/MultiBitECC/DDR4/2933MHz  \\
%         SSD           & Samsung SSD 870                \\
%         Network Card  & MT27800 Family [ConnectX-5]   \\
%         Network line   & Gigabit network cable         \\
%         Power         & PWS-2K08A-1R 2000w            \\
%         GPU           & NVIDIA A100 80GB PCIe         \\
%     \bottomrule
%   \end{tabular}
% \end{table}

\textbf{\textit{State-of-the-arts.}} We compare SparStencil with a wide range of state-of-the-arts, including cuDNN, AMOS, Brick, DRStencil, TCStencil, and ConvStencil in FP16. Moreover, considering that ConvStencil employs 3x temporal fusion for small kernels, SparStencil adopts the same approach in our comparative evaluation to ensure a fair performance comparison. 

% Notably, the design of ConvStencil is specifically tailored for computation on TCUs at FP64 precision. Considering that the size and shape of fragments at FP16 and FP64 differ on the TCU, we convert Constencil to FP16 precision with $16 \times 16 \times 16$ fragments. 
% 

\textbf{\textit{Benchmarks.}} We use various stencil kernels with different shapes as benchmarks. The specifics are detailed in Table \ref{tab:benchmark}~\cite{10.1145/2464996.2467268, yuansc17}.

\textbf{\textit{Metrics}}. Most works on stencil evaluate performance using GStencil/s~\cite{matsumura2020an5d, zhang2023perks, zhang2023revisiting, chen2019, yuanicpp19, yuansc17} (Gigastencils per second, denoting the number of stencil points updated per second) as a metric. We also adopt this metric, as defined in Equation \ref{eq:Gstencil}.
\begin{equation}
    \label{eq:Gstencil}
    \text{GStencil/s} = \frac{T \times \prod_{i=1}^{n} N_i}{t \times 10^9}
\end{equation}
where $T$ denotes the number of iterations, $n$ denotes the dimensionality of the stencil, $N_i$ denotes the size of the $i$-th dimension, and $t$ denotes the total execution time in seconds.

\begin{figure*}[htbp]
\includegraphics[width=1\textwidth]{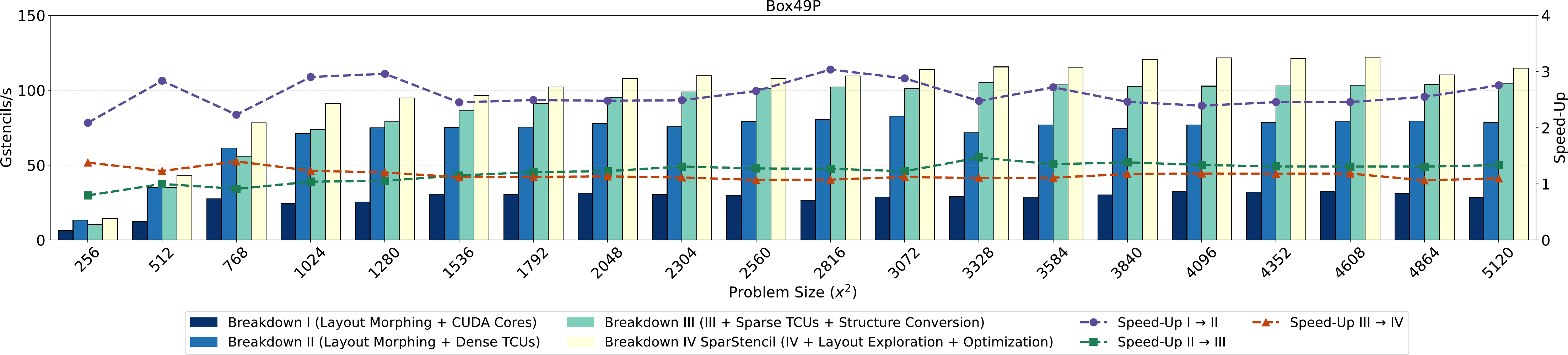} % Adjust width here
\caption{Performance breakdown of SparStencil}
\label{breakdown}
\end{figure*}
\begin{table}[htbp]
  \centering
  \caption{Configuration for Stencil Benchmarks}
  \label{tab:benchmark}
  \small
  \begin{tabular}{cccc}
    \toprule
    \textbf{Kernel} & \textbf{Points} & \textbf{Problem Size} & \textbf{Block} \\
    \midrule
        Heat-1D & 3 & $10240000 \times 10000$ & 1024 \\
        1D5P & 5 & $10240000 \times 10000$ & 1024 \\
        Heat-2D & 5 & $10240 \times 10240 \times 10240$ & $32 \times 64$ \\
        Box-2D9P & 9 & $10240 \times 10240 \times 10240$ & $32 \times 64$ \\
        Star-2D13P & 13 & $10240 \times 10240 \times 10240$ & $32 \times 64$ \\
        Box-2D49P & 49 & $10240 \times 10240 \times 10240$ & $32 \times 64$ \\
        Heat-3D & 7 & $1024 \times 1024 \times 1024 \times 1024$ & $8 \times 64$ \\
        Box-3D27P & 27 & $1024 \times 1024 \times 1024 \times 1024$ & $8 \times 64$ \\
    \bottomrule
  \end{tabular}
\end{table}

\subsection{Preprocessing Overhead Analysis ~\label{Sec:Preprocessing}} 

Figure \ref{overhead} illustrates the overhead associated with metadata, lookup table, and transformation across different stencils. 

Overall overhead is minimal and quickly amortized. For 1D stencils, a brief spike is observed (1D-5P: LUT peaks near 30\%) but rapidly diminishes. Higher-dimensional stencils, such as Box-3D27P, show consistently low and stable overhead due to favorable reuse patterns.

2D stencils (Box-2D9P, Box-2D49P) exhibit a slightly higher sustained overhead, occasionally exceeding 10\%. Compared to 1D stencils with simpler memory access and 3D stencils with higher data reuse in tensor cores, 2D stencils have a less optimal balance between sparsity transformation and computation-to-memory ratio, leading to relatively higher preprocessing costs.

\begin{figure}[htbp]
\includegraphics[width=0.5\textwidth]{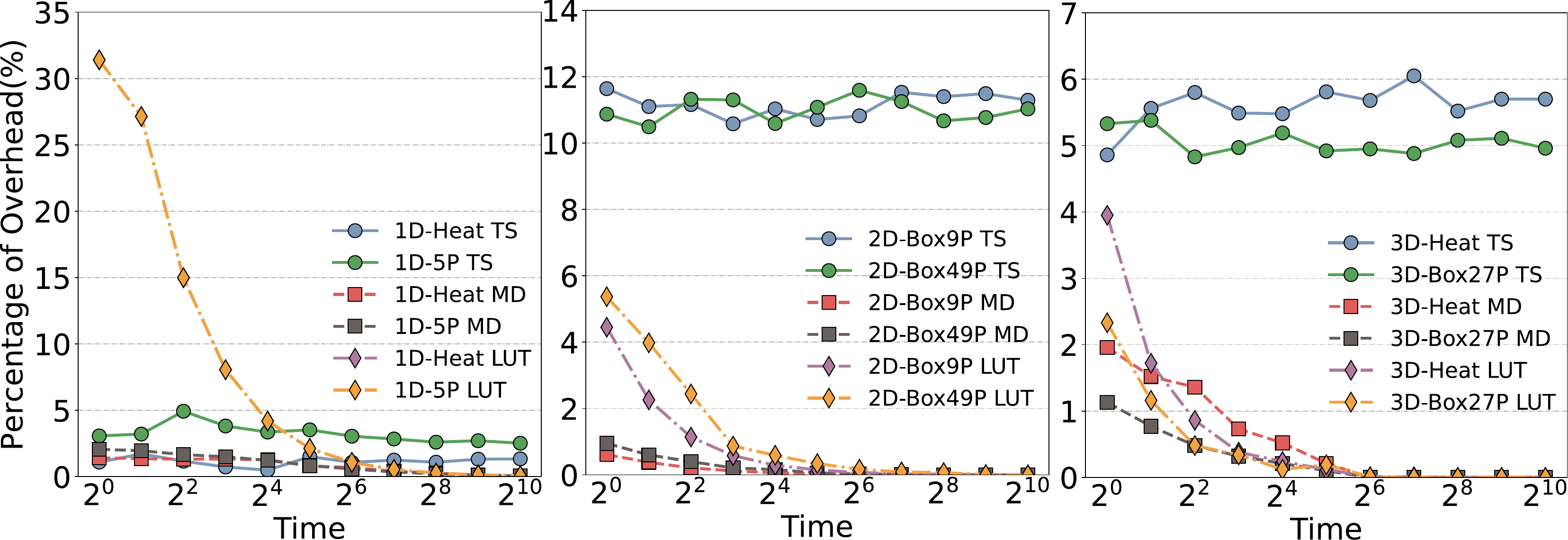} % Adjust width here
\caption{The percentages of overhead across various stencils. TS means Transformation, MD means Metadata and LUT means Lookup Table. }
% \vspace{-6.5mm}
\label{overhead}
\end{figure}

% \vspace{-0.1em}

\subsection{State-of-the-art Comparison}
\label{sec:sota}

Figure \ref{sota} compares SparStencil against state-of-the-art baselines. SparStencil consistently outperforms all reference methods, achieving 2.89x–60.35x speedup over cuDNN, which lacks Tensor Core support for stencil patterns and underperforms on one-channel convolutions. In contrast, AMOS falls short due to inefficient stencil-to-TCU mapping.

Against stencil-specialized baselines like DRStencil and ConvStencil, SparStencil retains a clear advantage. On Star-2D13P and Box-2D49P, it achieves up to 38.89\% higher performance, attributed to reduced memory traffic via layout search and superior compute throughput from sparse TCUs. This synergy of memory and compute underpins SparStencil's overall performance lead.

\subsection{Performance Breakdown}

Figure~\ref{breakdown} illustrates the incremental performance gains from each optimization stage using Box-2D49P as a representative case. Applying Layout Morphing on dense TCUs yields a 1.58x speedup over the CUDA baseline, attributed to the superior matrix multiplication capabilities of dense TCUs. Subsequently, employing PIT to harness sparse TCUs results in an additional 1.22x speedup, as sparse TCUs deliver double the computational power of dense TCUs. Notably, for problem sizes 256 and 768, PIT resulted in 0.79x and 0.90x performance, respectively, as the memory overhead of PIT exceeds the compute efficiency gains from sparse TCUs at limited scales. Finally, the application of further performance-boosting optimizations achieves an additional 1.24x improvement. With this, we have completed all optimizations in SparStencil.
\begin{figure}[htbp]
\includegraphics[width=0.5\textwidth]{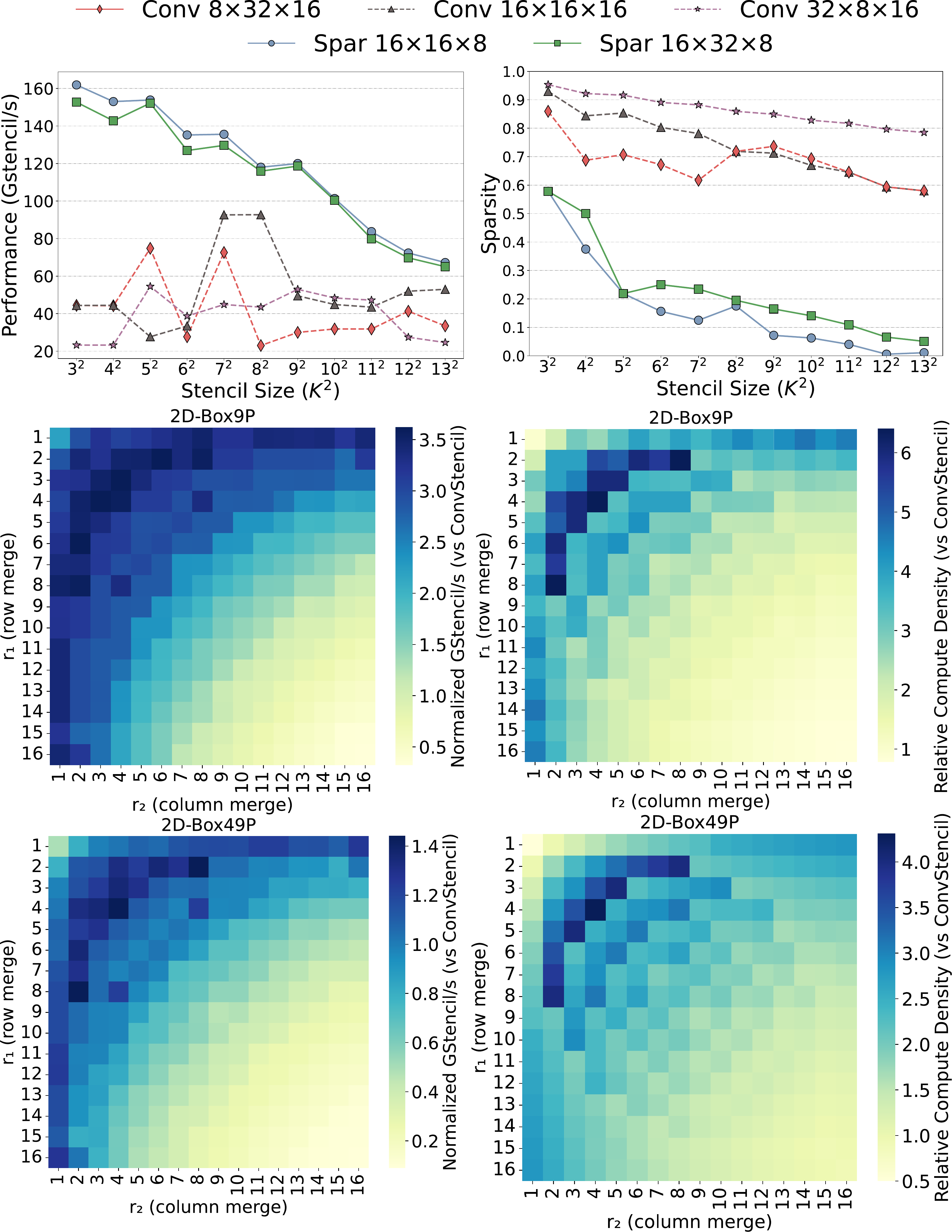} % Adjust width here
\caption{Performance, sparsity, and compute density across stencil sizes and layout configurations.
Top: Throughput (left) and sparsity ratio (right) across stencil sizes on different TCU fragments. Bottom: Performance and compute density over $(r_1, r_2)$ for two representative 2D stencil kernels.}
\label{adpative}
\end{figure}
\subsection{Adaptivity and Sparsity Across Varying Configurations }

\begin{figure*}[htbp]
\includegraphics[width=1\textwidth]{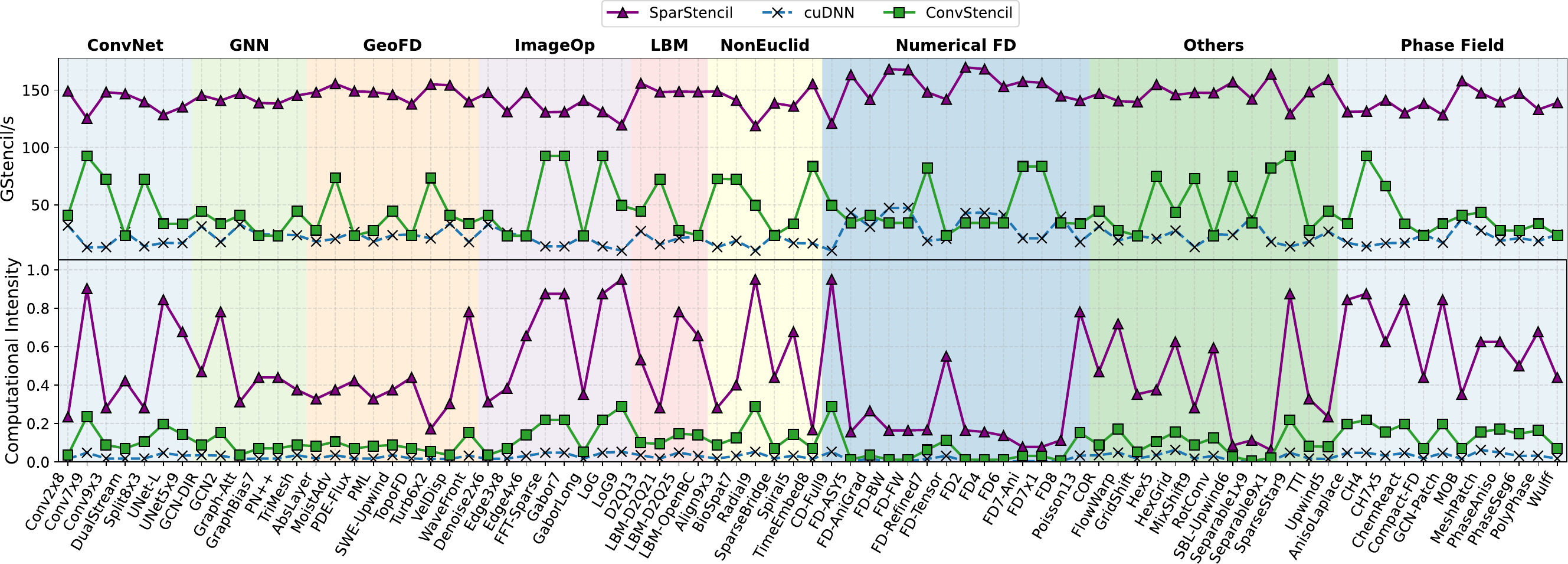} % Adjust width here
\caption{Performance and compute density breakdown across 79 real-world stencil kernels spanning 9 application domains. Top: end-to-end throughput (GStencil/s) of SparStencil, cuDNN, and ConvStencil; Bottom: corresponding compute intensity.}
\label{patterns}
\end{figure*}

To better evaluate adaptivity, we omit the 3x temporal fusion used in Section~\ref{sec:sota}, as it masks performance variation for small kernels. Here, we assess adaptivity and runtime sparsity across two dimensions: varying stencil sizes and diverse stencil patterns.

For varying stencil sizes, as illustrated in Figure \ref{adpative}, SparStencil consistently outperforms dense baselines across all stencil sizes and TCU fragments, achieving up to $4.1$× speedup and maintaining sparsity below $60\%$. Heatmaps show that optimal compute density emerges in specific $(r_1, r_2)$ regions, which are effectively discovered by our automated search. demonstrating the system’s ability to adaptively align diverse stencil patterns with TCU-efficient layouts.

For diverse stencil patterns, as illustrated in Figure \ref{patterns}, across 79 stencil kernels spanning 9 application domains, SparStencil consistently achieves higher throughput, reaching up to $156.7$ GStencil/s and outperforming cuDNN and ConvStencil by $6.3$x and $3.1$x on average, respectively. SparStencil maintains a stable performance, demonstrating robustness to stencil structural diversity. In terms of compute density, SparStencil achieves 17.92x and 4.46x over cuDNN and ConvStencil, respectively. Thanks to its adaptive layout search, SparStencil effectively reduces memory footprint and maintains stable performance even under extreme sparsity.

% \subsection{Preprocessing Overhead Analysis ~\label{Sec:Preprocessing}} 

\subsection{GPU Resource Utilization Comparison}

Figure~\ref{hardware} shows hardware utilization across six metrics. SparStencil achieves higher SM utilization (74.5\%) and occupancy (96.9\%) than ConvStencil (18.3\%, 61.3\%) and cuDNN (59.4\%, 88.5\%), enabled by sparsity-aware layout. L1/TEX cache and memory throughput reach 64.5\% and 64.1\%, indicating effective on-chip reuse. In contrast, DRAM throughput (17.5\%) and L2 cache activity (52.6\%) are lower than cuDNN (43.5\%, 61.6\%). This reduction reflects SparStencil’s layout-aware access patterns, which promote reuse in L1/shared memory, reducing dependence on L2 and minimizing global memory pressure.

% Figure \ref{hardware} provides a comprehensive comparison of compute and memory throughput across various stencil kernels, highlighting GPU hardware utilization metrics for cuDNN, ConvStencil, and SparStencil. Remarkably, SparStencil surpasses cuDNN---a state-of-the-art, expertly optimized library---and demonstrates an average improvement of $11.42\%$ in compute throughput and $19.34\%$ in memory throughput. Moreover, SparStencil achieves 1.62x higher L1-cache throughput compared to cuDNN and reduces execution time to an average of 0.51x that of cuDNN.

% These compelling results highlight the exceptional efficiency of SparStencil’s automatically generated kernels. By optimizing memory access patterns, enhancing data locality, and dynamically determining the optimal matrix layout based on the stencil input and TCU fragment, SparStencil maximizes the computation and memory throughput.

\begin{figure}[htbp]
\includegraphics[width=0.48\textwidth]{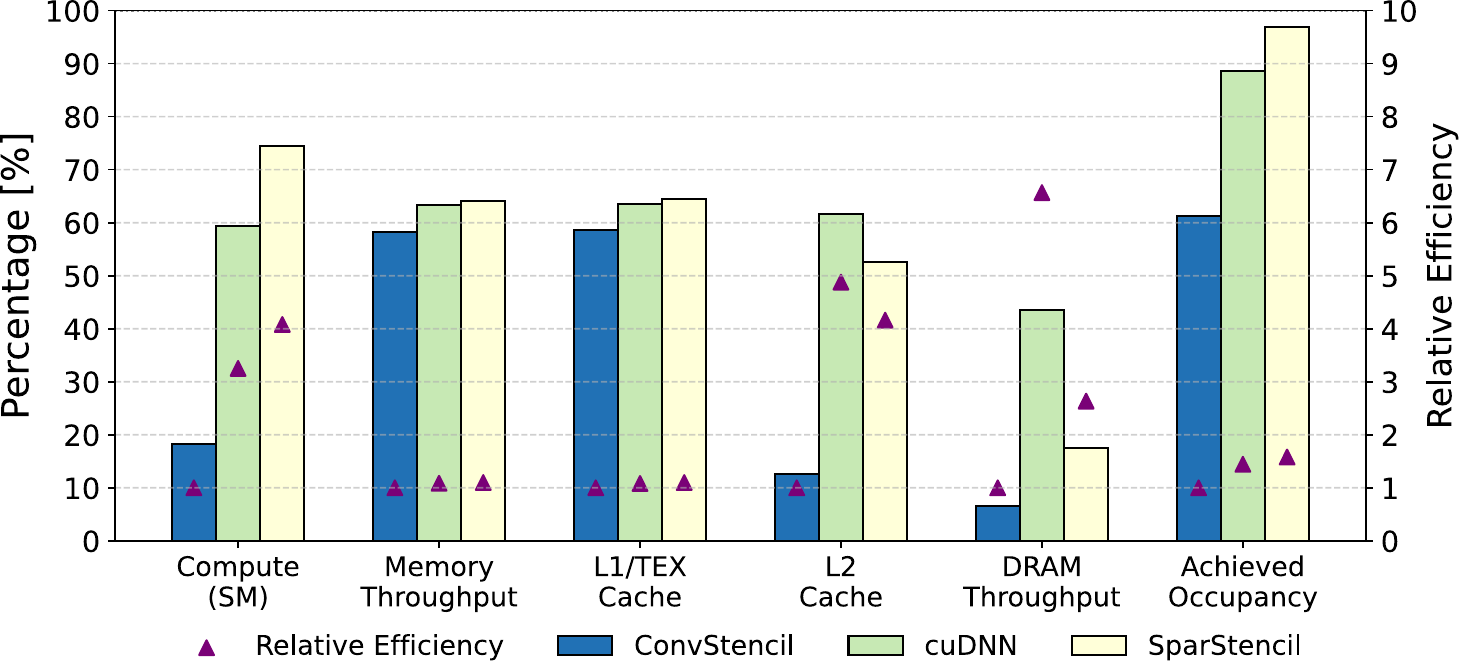} % Adjust width here
\caption{Hardware utilization comparison.}
\label{hardware}
\end{figure}

\begin{table}[b]
\footnotesize
\centering
\caption{Performance Comparison of SparStencil and SOTA methods under FP64 precision across different stencils (GFlops/s)}
\vspace{-8pt}
\label{tab:stencil_performance}
\begin{tabular}{lcccc}
\toprule
Method       & Heat-2D & Box-2D9P & Star-2D13P & Box-2D49P \\\midrule
AMOS         & 10.16   & 10.23    & 10.51      & 10.59     \\
cuDNN        & 64.33   & 64.57    & 17.05      & 17.15     \\
DRStencil    & 55.46   & 57.63    & 50.16      & 20.28     \\
ConvStencil  & 65.83   & 62.76    & 64.37      & 63.93     \\
SparStencil  & \textbf{72.49}   & \textbf{73.25}    & \textbf{71.34}      & \textbf{67.28}     \\\bottomrule
\end{tabular}
\end{table}

\subsection{Evaluating SparStencil Beyond FP16: Performance at FP64}

%Given current sparse TCU's limited support for high-precision computation, we conducted FP16-based validation SparStencil. To demonstrate our technique's extended applicability, we further evaluate against state-of-the-art approaches using FP64-enabled dense TCUs, as quantified in Table 
Given that current sparse TCUs lack direct support for FP64 computation, our primary evaluation mainly focuses on FP16 precision. To further evaluate its broader applicability, we extend our analysis to FP64-enabled dense TCUs, benchmarking against state-of-the-art approaches as shown in Table~\ref{tab:stencil_performance}.

Despite the absence of hardware-level sparsity acceleration, SparStencil consistently outperforms existing methods across all stencil configurations, achieving speedups ranging from 1.11x to 7.13x. This advantage stems from adaptive layout morphing and search strategy, which automatically restructures stencil matrices to optimize computational and memory throughput, ensuring high efficiency even without specialized sparse TCU acceleration.

These results validate two critical insights: (1) SparStencil remains effective and achieves performance leadership even on dense TCUs, demonstrating its robustness beyond FP16-restricted architectures. (2) Future sparse TCUs with FP64 support will further amplify SparStencil’s benefits, as our sparse-aware optimization framework is inherently aligned with next-generation hardware trends.

%Notably, despite losing hardware-level sparsity acceleration in dense TCUs, SparStencil still demonstrates consistent performance gains across all stencil configurations, achieving average 1.11-7.13x speedup over the state-of-the-art approaches in FP64. This is attributed to our adaptive layout morphing strategy, which automatically optimizes stencil matrix layouts to maximize computation and memory throughput.

%These results validate two key advantages: (1) Our method maintains compatibility and performance leadership even without hardware-level sparsity acceleration in current TCU architectures, and (2) Future TCU designs with enhanced sparsity support will further amplify our technique's efficiency potential, as our sparse-aware optimization framework is fundamentally aligned with emerging hardware trends. 

% Our SparStencil implementation demonstrates consistent performance advantages across all stencil configurations, achieving 1.10-1.45× speedup over the second-best ConvStencil method. Particularly notable is the $16.7\%$ improvement on the Box-2D9P pattern and $10.8\%$ on Star-2D13P, highlighting our solution's effectiveness with complex neighborhood accesses. The results maintain computational superiority even on the 49-point Box-2D49P stencil (5.2\% faster than ConvStencil), suggesting strong scalability to high-order patterns. This FP64 validation confirms our methodology's fundamental advantages transcend numerical precision constraints, with particularly pronounced benefits for memory-bound operations characteristic of modern tensor architectures.

\section{Related Work}

The optimization of stencil computations has been extensively studied on both CPUs~\cite{10.1145/3458817.3476154, yuansc17} and GPUs~\cite{8451874, 6114468, 10.1145/2491956.2462176}. 

On CPUs, vectorization~\cite{10.1145/3458817.3476154, 10.1145/2464996.2467268, 10.1007/978-3-642-19861-8_13, 10579118}, data reuse strategies~\cite{8665800, 10.1145/2666356.2594342, zhaosc19}, and tiling techniques, such as diamond~\cite{10.1109/SC.2012.107, 10.1145/1375581.1375595}, time skewing~\cite{1592833, 10.1023/A:1015460304860}, rectangular~\cite{1592745}, and tessellating tiling~\cite{yuansc17}, are key methods for improving performance. 

On GPUs, spatial tiling~\cite{7034720, zhaosc19, maruyama2014optimizing}, temporal tiling~\cite{Holewinski12ics, 10.1145/1542275.1542313, 10.1145/2830018.2830025, 10.1145/2400682.2400713, 7582549, 10.1145/2458523.2458526}, unrolling~\cite{10.1145/3469030}, prefetching~\cite{8820786}, and streaming~\cite{8451874} techniques effectively leverage GPU parallelism and memory hierarchies. Frameworks like Brick~\cite{zhaosc19, 10.1145/3437801.3441598, 8639931} exploit fine-grained data reuse, while DRStencil~\cite{you2021drstencil} applies fusion-partition optimization for stencil acceleration. Recent research also explores Tensor Cores, with TCStencil~\cite{liu2022toward} applying them directly, ConvStencil~\cite{10.1145/3627535.3638476} bridging convolution and stencil computations, FlashFFT~\cite{10.1145/3710848.3710897} restructures FFT-based stencil computations on TCU and LoraStencil~\cite{SC41406.2024.00059} reduces redundancy through low-rank decomposition of stencil. Furthermore, libraries like cuDNN~\cite{cudnn, chetlur2014cudnn} optimize stencil-related operations, and AMOS~\cite{zheng2022amos} maps depth-wise convolutions, equivalent to stencil operations, to hardware units, including Tensor Cores.

\section{Conclusion}

This paper presents SparStencil, a novel system leveraging sparse Tensor Cores for stencil computations via Permutation Invariant Transformation. It integrates Adaptive Layout Morphing, Structured Sparsity Conversion, and Automatic Kernel Generation to optimize stencil computations on sparse TCUs. Evaluations demonstrate its effectiveness, outperforming state-of-the-art approaches, with potential benefits for scientific and engineering applications.

\newpage

\bibliographystyle{ACM-Reference-Format}
\bibliography{main}

\end{document}